%% file: OpenGTsynch-Main.tex
\documentclass[copyright,creativecommons]{eptcs}
\providecommand{\event}{ICE 219} 
\usepackage{breakurl}             
\usepackage{underscore}           

\usepackage[algoruled,vlined]{algorithm2e}
\SetKwProg{Fn}{Function}{}{}

\usepackage[nomargin,inline,index,%
]{fixme}

\usepackage[shortlabels]{enumitem}
\usepackage{makeidx}
\usepackage{color}
\usepackage{wrapfig}

\usepackage{mathpartir,amsmath,amssymb
}

\usepackage{prooftree}
\usepackage{mathtools}

\usepackage{hyperref}
\usepackage[capitalise]{cleveref}

\usepackage[all]{xy}

\input{macrosDS}

\input{macros}

\begin{document}

\title{Open Multiparty Sessions} 

\author{Franco Barbanera\thanks{
Partially supported by the 
project ``Piano Triennale Ricerca'' DMI-Universit\`a di Catania. 
}\\
Dipartimento di Matematica e Informatica\\
Universit\`a di  Catania, Catania Italy\\
barba@dmi.unict.it\and
Mariangiola Dezani-Ciancaglini\thanks{Partially supported by Ateneo/Compagnia di San Paolo 2016/2018 project ``MnemoComputing - Components for Processing In Memory".}\\
Dipartimento di Informatica,\\
Universit\`a di Torino, Torino Italy\\
dezani@di.unito.it
}


\def\titlerunning{Open Multiparty Sessions} 

\def\authorrunning{
 Barbanera \& Dezani-Ciancaglini 
}

\maketitle

\begin{abstract} 
Multiparty sessions are systems of concurrent processes, which allow several participants 
to communicate by sending and receiving messages. Their overall behaviour can be described
by means of global types.
Typable multiparty session  enjoy lock-freedom.  
We  look at
multiparty sessions as {\em open} systems
by a suitable definition of {\em connection} transforming {\em compatible}
processes into {\em gateways} (forwarders). 
A relation resembling the standard subtyping relation for session types is used to formalise compatibility.
We show that the session obtained by connection  can be typed by manipulating the global types of the starting sessions. This allows us to prove that  
lock-freedom is preserved by connection. 
\end{abstract}

\input{OpenGTsynch-intro}

\input{mpstN}

\input{mergeN}

\input{OpenGTsynch-discussion}

\paragraph{Acknowledgments} We are indebted to Ivan Lanese and Emilio Tuosto for many enlightening discussions on the subject of this paper. We gratefully  acknowledge 
the anonymous referees for the interaction through the ICE web site 
and for their reports. The final version of this paper strongly improved in clarity and correctness thanks to their observations and suggestions.

\bibliographystyle{eptcs}
\bibliography{main}

\end{document}

%% file: macrosDS.tex
\newtheorem{theorem}{Theorem}[section]
\newtheorem{lemma}[theorem]{Lemma}
\newtheorem{example}[theorem]{Example}
\newtheorem{definition}[theorem]{Definition}
\newtheorem{proposition}[theorem]{Proposition}
\newtheorem{remark}[theorem]{Remark}
\newenvironment{proof}{{\em Proof.}}{
}
\newenvironment{myproof}
               {
               \begin{proof}
               }{
              \end{proof}\vspace{13pt}
              }

\newcommand{\QED}{{\hfill $\square$}}

\newcommand{\ok}{ok}

\newcommand{\from}{\leftarrow}

\newcommand{\pre}{pre-}
\newcommand{\pair}[2]{ ( #1,#2) }

\newcommand{\T}{{P}}
\newcommand{\R}{R}
\renewcommand{\S}{S}

\newcommand{\PP}{\ensuremath{P}}
\newcommand{\Q}{\ensuremath{Q}}

\newcommand{\Lab}{\Lambda}

\newcommand{\Labgt}{\Gamma}
\newcommand{\Labt}{\Labgt}

\newcommand{\rulename}[1]{\mbox{\small[\textsc{#1}]}}

\newcommand{\depth}[2]{ weight (#1, #2)}

\newcommand{\mylabel}[1]{\label{#1}}

\newcommand{\mergeG}{\text{\sc cn}}

\newcommand{\s}{{\sf s}}
\newcommand{\pp}{\mathsf{p}}

\newcommand{\q}{\pq}
\newcommand{\pq}{{\sf q}}
\newcommand{\pr}{{\sf r}}
\newcommand{\ps}{{\sf s}}
\newcommand{\pt}{{\sf t}}

\newcommand{\val}{\kf{v}}

\newcommand{\sep}{\ensuremath{~~\mathbf{|\!\!|}~~ }}
\newcommand{\kf}[1]{\ensuremath{\mathsf{#1}}}
\newcommand{\pc}{\ensuremath{~|~}}

\newcommand{\G}{\ensuremath{{\sf G}}}
\newcommand{\Y}{\ensuremath{{\sf Y}}}
\newcommand{\Gvti}[5]{\ensuremath{#1\to#2:\{#3_i. #5_i \mid {\scriptstyle 1 \leq i \leq n } \}}}

\newcommand{\mkt}[3]{\ensuremath{#1\to#2:#3  }}

\newcommand{\N}{\ensuremath{\mathcal M}}
\newcommand{\M}{\ensuremath{\mathcal M}}

\newcommand{\pa}[2] 
{#1 \triangleright  #2}

\newcommand{\der}[3]{ #1 \vdash   #2  :#3}
\newcommand{\derp}[3]{ #1 \vdash^+   #2  :#3}

\newcommand{\participantP}[1]{\mathtt{ptp} (#1)}
\newcommand{\participantS}[1]{\mathtt{pts} (#1)}
\newcommand{\participantG}[1]{\mathtt{ptg} (#1)}
\newcommand{\labels}[1]{ {msg} (#1) }
\newcommand{\paths}[1]{{paths} (#1)}
\renewcommand{\path}{\mypath}
\newcommand{\mypath}{\rho}
\newcommand{\emptypath}{\epsilon}

\newcommand{\project}{\upharpoonright}

\newcommand{\proj}[2] 
{ #1 \! \! \project_{#2} }
\newcommand{\projp}[2] 
{ (#1) \project_{#2}}

\newcommand{\pin}[2]{#1 ? #2}
\newcommand{\pout}[2]{#1 ! #2}

\newcommand{\mklab}[2]{#1 \uplus #2}
\newcommand{\setl}[3]{\{#1_{i}. #2_{i} \mid {\scriptstyle 1 \leq i \leq #3} \}}

\newcommand{\setlp}[4]{\{#1_{i}. \proj{#2_{i}}{#4} \mid {\scriptstyle 1 \leq i \leq #3} \}}

\newcommand{\inact}{\ensuremath{\mathbf{0}}}

\newcommand{\tend}{\mathtt{end}}

\renewcommand{\S}{S}

\newcommand{\subt}{\leqslant}
\newcommand{\subtp}{\leqslant^+}

\newcommand{\suptp}{\,^+\!\!\!\geqslant}

\newcommand{\red}{\longrightarrow}

\newcommand{\redparr}[1]{\xrightarrow{#1}}
\newcommand{\redpar}{\longrightarrow}

\newcommand{\cinferrule}[3][]{
  \mprset{fraction={===},
  fractionaboveskip=0.2ex,
  fractionbelowskip=0.4ex}
  \inferrule[#1]{#2}{#3}
}

\definecolor{ceca}{rgb}{1,0.5,0}

   \newcommand{\mysection}[1]{
  \section{#1}
  }

%% file: macros.tex
\DeclareMathAlphabet{\mathpzc}{OT1}{pzc}{m}{it}

\newcommand{\vr}{\iota}

\newcommand{\set}[1]{\{#1\}}

\newcommand{\Comment}[1]{ }

\newcommand{\ByDef}{\triangleq}
\def\Pred[#1]{~[\,#1\,]}

\newcommand{\eraseP}{\mbox{\sc ers}}

\newcommand{\HH}{\mathsf{h}}
\newcommand{\KK}{\mathsf{k}}

\newcommand{\interfacecomp}{\!\leftrightarrow\!}

\newcommand{\conn}[2]{\stackrel{\hspace{-2.5pt}#1\!\leftrightarrow\!#2\!\!\hspace{-0.5pt}}{}}

\newcommand{\noint}{\texttt{\#}}

\newcommand{\gateway}[1]{\mathsf{gw}(#1)}

\newcommand{\Dual}[1]{\overline{#1}}

\newcommand{\Set}[1]{\{#1\}}



%% file: OpenGTsynch-intro.tex

\mysection{Introduction}
Distributed systems are seldom developed as independent
entities and, either directly in their design phase or even after their
deployment, they should be considered as open entities ready for interaction
with an environment. 
In general, it is fairly natural to expect to connect open systems
as if they were composable modules, and in doing that we should rely on
``safe'' methodologies and techniques, guaranteeing 
the composition not to ``break'' any relevant property of the 
single systems. 

In \cite{BdLH18} a methodology 
has been proposed for the connection 
of open systems, consisting in replacing
{\em any} two participants - if their behaviours are ``compatible'' -
by two  forwarders, dubbed {\em gateways}, enabling the systems to interact. 
The  behaviour of any participant can be looked at
as an {\em interface}
since, without 
loss of generality, 
the notion of interface is interpreted not
as the description of the interactions ``offered'' by a system
but, dually, as those ``required''  by a possible environment (usually  
another system).

 Inspired by \cite{LTPC}, the aim of the present paper is to look for a choreography formalism enabling to ``lift'' the connection-by-gateways
by means of a proper function definable on  protocol descriptions. 
The function should yield  the  protocol  
of the system obtained by connecting the systems described by the arguments of the function itself.
 Connected systems would hence enjoy all the good communication properties
guaranteed by the formalism itself, which - with no ad-hoc extension of the syntax - could be  seen
as a choreography formalism for open  systems.
 
We took into account the choreography model of MultiParty Session Types (MPST)
\cite{HYC08,Honda2016}.
Of course not all of the MPST formalisms are suitable for our aim. For instance, in the formalisn of \cite{DBLP:journals/mscs/CoppoDYP16} the requirements 
imposed by its  type system are too strong for the 
 gateway processes to be typed, so preventing the function we are looking for to be definable. 
 
The MPST formalism  that we introduce in the present paper (inspired by \cite{DS}) proved to be  a 
right candidate. 
The simplicity of the calculus allows to get rid of channels and local types.
Moreover, its abstract point of view for what concerns global and local behaviours (looked at as infinite regular trees) also enables to focus on the relevant  
aspects of the investigation without the hindering syntactic descriptions of recursion.
In particular, with respect to \cite{DS}, we relax the conditions imposed on global types in order to be projectable, so ensuring projectability of ``connected'' global types.
 (a property that  does  not hold in the formalism of \cite{DS}).
In our  formalism, typable
systems are guaranteed to be lock-free  \cite{Kobayashi02}.  
 The systems obtained by connecting typable systems are lock-free too. The main tool is a function from the global types of the original systems to the global type of the system obtained by connection. \\
In the present setting  it is also possible to investigate in a clean way the notion of {\em interface compatibility}:
we show that the compatibility relation used in \cite{BdLH18} can be relaxed to a relation closely
connected to the observational preorder of \cite{DS}, in turn corresponding to the subtyping  relation for session types of \cite{GH05,DemangeonH11}.

{\bf Outline} The first three sections introduce our calculus of multiparty sessions, together with their global types, and prove the properties of well-typed sessions. In the following two sections  we define the compatibility relations  and the gateway connections for sessions and global types, respectively. Our main result, i.e.  the typability and hence the lock-freedom of the session obtained by 
gateway connection, is Theorem \ref{t:m}. Sections \ref{s:rfw} and \ref{s:fwc} conclude discussing related and future works, respectively.

%% file: mpstN.tex

 \mysection{Processes and Multiparty Sessions}

We use the following base sets and notation: 
 \emph{messages}, ranged over by $\ell,\ell',\dots$;
 \emph{session participants}, ranged over by $\pp,\pq,\ldots$;
\emph{processes}, ranged over by $P,Q,\dots$; 
\emph{multiparty sessions}, ranged over by 
\linebreak 
$\N,\N',\dots$;
\emph{integers}, ranged over by $n, m, i, j, ,\dots$.

Processes  implement the behaviours of single participants.
The input process $\pin{\pp}{\setl{\ell}{\PP}{n}}$ 
 waits for one of the  messages $\ell_i$ from participant $\pp$; 
the output process $\pout{\pp}{\setl{\ell}{\PP}{n}}$ 
chooses one message $\ell_i$ and sends it to participant $\pp$. 
We use $\Lab$ as shorthand for $\setl{\ell}{\PP}{n}$.
  We define   the multiset of messages in $\Lambda$ as
$\labels {\setl \ell \PP  n}   = \{ \ell_i \mid 1 \leq i \leq n \}$. 
 After sending or receiving the message $\ell_i$ the process reduces to $\PP_i$ ($1\leq i\leq n$). 
The set $\Lab$ in $\pin{\pp}\Lab$ acts as an {\em external choice},
while the same set in $\pout{\pp}\Lab$ acts as an {\em internal choice}.
 In a full-fledged calculus, messages would carry values, namely
  they would be of the form $\ell(\val)$. 
  Here for simplicity we consider only pure messages.  This agrees with the focus of session calculi, which is on process interactions that do not depend on actual transmitted values.  \\
For the sake of abstraction, we do  not  take into account any explicit syntax for recursion, but rather consider
processes as, possibly infinite, regular trees.

It is handy to first define \pre processes, since the processes must satisfy conditions which can be easily
 given using the tree representation of \pre processes. 

\begin{definition}[Processes]
\begin{enumerate}[label=(\roman*)]
\item
We say that $\PP$ is   a {\em \pre process}  and
$\Lab$ is a {\em \pre choice of messages}  if they  are generated by the
 grammar:  \\
\centerline{$
\begin{array}{lll@{\quad\qquad\qquad}lll}
\PP   &   ::=^{coinductive}  & 
 \inact  \quad \sep \quad 
\pin{\pp}{\Lab}   \quad   \sep   \quad   
\pout{\pp}{\Lab}   
&
\Lab & ::= {\setl{\ell}{\PP}{n}} 
                   \end{array}
$}
and 
  all  messages in  $\labels {\Lab}$  are pairwise distinct.

 \item 
 The \emph{tree representation} of   a  \pre process   
is  a 
  directed rooted  tree,  where:\ 
 (a)    each  internal 
  node is labelled by   
$\pp ?$ or $\pp !$
 and has  
as many children as the number of messages,\ 
(b) the edge from 
 $\pp ?$ or $\pp !$ to the child $\PP_i$ is labelled by $\ell_i$ and \ 
(c) the  leaves
of the tree (if any)
are labelled by $\inact$.
 \item
 We say that a \pre process $\PP$ is a {\em process} if 
the tree representation of $\PP$ 
is regular (namely, it has finitely many distinct sub-trees).
We say that  a pre-choice of messages 
$\Lab$ is  a {\em choice of messages}
if all the \pre processes in $\Lab$ are processes.
\end{enumerate}
 \end{definition}
 
  We identify   
 processes  
with their 
tree representations and  we shall sometimes refer to the trees
as the processes themselves. 
 The regularity condition  implies that we only consider processes 
admitting a finite description. 
 This  
is equivalent to  writing  
processes 
with  $\mu$-notation and  an  
equality
which allows for an infinite number of unfoldings.
This is also called the {\em  equirecursive approach}, since it views
processes as the unique solutions 
of (guarded) recursive equations~\cite[Section 20.2]{pier02}.
The existence and
uniqueness of a solution   
 follow from known
results~(see \cite{Courcelle83}  and 
 \cite[Theorem 7.5.34]{BookCC}). 
It is natural to use {\em coinduction} as the main 
logical tool, as we do in most of the proofs.
In particular, we adopt the coinduction style
advocated in \cite{KozenS17} which, without any loss of formal rigour,
 promotes readability and conciseness.

 We define  
the set $\participantP\PP$ of participants of process $\PP$ by: $
\participantP\inact   =\emptyset$ and\\
\centerline{$
\participantP{\pin{\pp}{\setl{\ell}{\PP}{n}}}  =
\participantP{\pout{\pp} \setl{\ell}{\PP}{n}}
 = \{ \pp \} \cup \participantP {\PP_1} \cup \ldots \cup \participantP {\PP_n}
$}
 The regularity of processes  
 assures that the set of  participants  is finite.
 
  We shall write 
  $\mklab{\ell. \PP}{ \Lab}$ for $\{\ell. \PP\}\! \cup\! \Lab$ if 
 $\ell\!\not \in\! \labels \Lab$ and 
 $\mklab {\Lab_1}{\Lab_2}$ for $\Lab_1 \!\cup\! \Lab_2 $ if $\labels {\Lab_1} \!\cap\! \labels {\Lab_2} = \emptyset$. 
  We
shall also 
omit  curly  brackets in choices with 
   only one branch and trailing $\inact$ processes.
 
 \medskip
 A  {\em multiparty session}  is the parallel composition of pairs participants/processes. 
\begin{definition}[Multiparty Sessions]
A {\em multiparty session}
 $\N$ is defined by the following grammar:\\
\centerline{$
\begin{array}{lll}\N  &  ::=^{inductive}  & \pa\pp\PP  \quad \sep  \quad \N \pc\N \end{array}
$}
and it should 
satisfy the following conditions:\\
(a) 
 In  $\pa{\pp_1}{\PP_1}\pc \ldots\pc \pa{\pp_n}{\PP_n}$   
all the $\pp_i$'s ($1\leq i\leq n$) are distinct; \\
(b)  In $\pa{\pp}{\PP}$  we require $\pp\not\in\participantP\PP$ (we do not allow self-communication).
\end{definition}
We shall use 
   $\prod\limits_{1\leq i\leq  n} \pa{\pp_i}{\PP_i}$ as shorthand  for $\pa{\pp_1}{\PP_1}\pc \ldots\pc \pa{\pp_n}{\PP_n}$. \\ We define $\participantS{\pa\pp\PP}=\set\pp$ and $\participantS{\N\pc\N'}=\participantS{\N}\cup\participantS{\N'}$.

\paragraph*{Operational Semantics}

The structural congruence $\equiv$ between two
multiparty sessions  establishes that 
 parallel composition is commutative, associative and has neutral elements $\pa\pp{\inact}$ for any fresh $\pp$.

 The reduction for multiparty 
sessions allows participants to choose and
 communicate messages. 
 
\begin{definition}[LTS for Multiparty Sessions]\mylabel{slts}
The {\em   labelled transition   system   (LTS) 
for multiparty sessions}  
  is the closure under structural congruence of the reduction specified by the unique rule:\\
  \centerline{$
 \inferrule[\rulename{comm}]{  \labels{\Lab} \subseteq \labels{\Lab'} 
     }{
    \pa\pp  {\pout{\q}{  (\mklab{\ell. \PP}{\Lab} )  }}
    \;
     \pc \;
     \pa \pq {\pin{\pp}{  (\mklab{\ell. \Q}{\Lab'} )  }  }\pc\N
    \redparr{\pp \ell \pq}
    \pa\pp{\PP}\;\pc\;\pa\q\Q\pc\N
    }$}
\end{definition}
Rule \rulename{comm} makes  the communication  possible: 
 participant $\pp$ sends  message 
 $\ell$  to   participant $\q$. 
  This rule is   non-deterministic in the choice of messages.  
The condition $\labels{\Lab} \subseteq \labels{\Lab'}$
assures that  the sender  
can freely choose the message, 
since the receiver must offer  all sender messages 
and  possibly 
 more.  This allows us to distinguish in the operational semantics between internal and external choices. 
We  use
$\M \redparr{ \lambda } \M'$ as shorthand for $\M \redparr{ \pp \ell \q  } \M'$.
We sometimes omit the label writing    $\redpar$. 
As usual, $\red^*$ denotes  the reflexive and transitive closure of $\red$.

\begin{example}
 \mylabel{ex:M}
 {\em 
 Let us consider a system (inspired by a similar one in \cite{BdLH18}) with participants $\pp$, $\q$, and $\HH$ interacting according the following protocol.  Participant  $\pp$ keeps on sending text messages to q, which has to deliver them to  $\HH$.
After a message has been sent by  $\pp$, the next one can be sent only if the previous has been received by $\HH$ and its propriety of language  ascertained,  i.e if it does not contain, say, rude or offensive words. 
Participant $\HH$ acknowledges to $\q$ the propriety of language of a received text by means of the message $\textit{ack}$. In such a case $\q$ sends to $\pp$ an $\textit{ok}$ message so that $\pp$ can proceed by sending a further message. More precisely: 
\begin{enumerate}
\item 
 $\pp$ sends a text message to $\q$ in order to be delivered to $\HH$, which accepts only texts
possessing a good propriety of language;
\item  then $\HH$  either 
          \begin{enumerate}
              \item  sends an $\textit{ack}$ to $\q$ certifying the reception of the text and its propriety. In this case
                       $\q$ sends back to $\pp$ an $\textit{ok}$ message  and the protocol goes back to 1.,   so that  $\pp$ can proceed by sending a further text
                        message;    
           \item   sends  a $\textit{nack}$ message to inform $\q$ that the text has not the required propriety of 
                    language. In such a case $\q$ produces 
 \textit{transf} (a semantically invariant reformulation  of the text),  
                    sends it back to $\HH$  and the protocol goes to 2. again. 
                    Before doing that, $\q$ informs  $\pp$ (through the \textit{notyet} message) that 
                    the  text   has not been accepted yet and a reformulation has been requested;
           \item  sends  a $\textit{stop}$ message to inform $\q$ that no more text will be accepted. In such a case
                    $\q$ informs of that also $\pp$.
 \end{enumerate}
 
 \end{enumerate}
A multiparty session implementing this protocol is: $
\M= \pa \pp  \PP    \pc   \pa \q \Q     \pc      \pa \HH H  $
where\\ 
\centerline{
$\begin{array}{l@{\qquad \qquad}l}
\PP=\pout {\q}{\textit{text}}.P_1   & P_1=\pin {\q}{\set{\textit{ok}.P,\textit{notyet}.P_1,\textit{stop}}}\\
\Q= \pp?\textit{text}.\HH! \textit{text}. Q_1 &
 Q_1=\pin {\HH}{\set{\textit{ack}.\pout {\pp}{\textit{ok}}.Q,\textit{nack}.\pout {\pp}{\textit{notyet}}.\pout {\HH}{{\textit{transf}}.Q_1, \textit{stop}.\pp!\textit{stop}}}}\\
  H = \pin {\q}{\textit{text}}. H_1  &   H_1 =  \pout{\q}{\{ \textit{ack}.H, \textit{nack}.\pin {\q}{\textit{transf}.H_1}, \textit{stop} \}}
  \end{array}$
  }
}
\end{example}

We end this section by defining the property of 
lock-freedom for multiparty session as in  \cite{Kobayashi02}. 
Lock-freedom ensures both progress and no starvation (under fairness assumption). I.e. it  guarantees the 
absence of deadlock and that all participants willing to communicate can do it. Recall that $\pa\pp{\inact}$ is the neutral element of parallel composition.
\begin{definition}[
Lock-Freedom]\mylabel{d:lf}
We say that a multiparty session \emph{ $\M$ is a   
lock-free session} if
\begin{enumerate}[label=(\alph*)]   
\item\mylabel{d:lf1} 
$\M\red^*\M'$ implies either $\M'\equiv \pa\pp{\inact}$ or $\M'\redpar\M''$,  and 
\item\mylabel{d:lf2} 
 $\M\red^*\pa\pp\PP\pc\M'$ and $\PP\not=\inact$ imply $\pa\pp\PP\pc\M'\red^*\M''\redparr\lambda$ and $\pp$ occurs in $\lambda$.

\end{enumerate}
\end{definition}

\mysection{Global Types and Typing System}

The behaviour of multiparty sessions can be disciplined by means of types, as usual.  Global types describe the whole conversation scenarios of multiparty sessions.  As in \cite{DS} we
 directly  assign global types to
multiparty sessions without the usual  detour around  
session types 
and subtyping  \cite{HYC08,Honda2016}.

The type $\Gvti \pp \q \ell \S {\G} $ formalises a protocol where participant
$\pp$ must send  
to $\q$ a 
message $\ell_i$ for some $1 \leq i \leq n$ and
then, depending  on  which $\ell_i$ was chosen by $\pp$, the protocol
continues as $\G_i$.  We use $\Labgt$ as shorthand for $\setl {\ell}{\G} {n}$  and define  the multiset 
$\labels{\setl {\ell}{\G} {n}} = \Set{\ell_i\mid 1\leq i \leq n}$.  
As for processes, we define first \pre global types and then global types.

\begin{definition}[Global Types]
  \begin{enumerate}[label=(\roman*)]
  \item
We say that $\G$ is a {\em \pre global type}  and  $\Labgt$
is a {\em \pre choice of communications}  
if they are generated by 
the grammar: \\[1mm]
\centerline{$
 \G ::=^{coinductive} \  \tend \quad  \sep  \quad \mkt \pp \q \Labgt \
 \qquad\qquad  \Labgt  :=  \setl {\ell}{\G} {n}
$}
and  
 all messages in  $\labels{\Labgt}$   are pairwise distinct.
 \item The {\em tree representation} of a   \pre global type  
is built as follows:\
(a) each internal node is labelled by  $\pp \to \q$ 
and has as many children as the number of messages,\
(b) the edge from 
 $\pp \to \q$ to the child $\G_i$ is labelled by $\ell_i$ and \
(c) the leaves of the tree (if any) are labelled by $\tend$.
 \item
We say that  a \pre global type 
$\G$ is a {\em global type} if 
the tree representation of $\G$ is {\em regular}.
We say that a  \pre choice of communications 
$\Labgt$ is a {\em choice of communications}  
if all the 
\pre global types in $\Labgt$ are  global types. 
\end{enumerate}
\end{definition}

We identify  \pre global types  and global types with their 
tree representations 
and  we  shall  sometimes refer to the tree representation 
as the global types themselves. 
As for processes,  
the regularity condition
 implies that we only consider global types 
admitting a finite representation.

 The set $\participantG\G$ of participants of global type $\G$
is defined similarly to  that   
of 
processes. 
 The regularity of  global types  
 assures that the set of participants is finite.
 We  shall  write 
 $\mklab{\ell. \G}{ \Labgt}$ for $\{\ell. \G \} \cup \Labt$ if 
 $\ell \not \in \labels\Labgt$ and 
 $\mklab {\Labgt_1}{\Labgt_2}$ for $\Labgt_1 \cup \Labgt_2 $ if $\labels{\Labgt_1} \cap \labels{\Labgt_2} = \emptyset$. 
We  shall  omit  curly  
brackets in  choices  with only one branch  and trailing $\tend$s. 
 
 Since all messages  in communication choices  
are pairwise distinct, the set of  paths  in the trees representing  global types 
are determined by  the labels of nodes and edges  
found on the way, omitting  the leaf label $\tend$.   Let  $\path$ range over paths of global types. 
Formally  the set of paths of a global type can be defined as a set of sequences
as follows ($\emptypath$ is the empty sequence): \\
\centerline{$
\paths\tend = \{ \emptypath \}
\qquad\quad
\paths{\Gvti \pp \q \ell \S {\G} }
= \bigcup_{1 \leq i \leq n} \{ \pp \to \q  \, \ell_i \, \path \mid \path \in \paths {\G_i} \}
$}

\noindent
 Note that every infinite path  of   a global type 
 has infinitely many occurrences of $\rightarrow$.
 
 \begin{example}
 {\em
 \mylabel{ex:G}
A global type representing the protocol of Example \ref{ex:M} is:\\[1mm]
\centerline{$
\begin{array}{lcll}
\G  &=&  & \pp \to \q :   \textit{text} .  \q \to \HH  : \textit{text} .\G_1\\[1mm]
\G_1 &= &
& \HH \to \q : \{ \textit{ack} :    \q \to \pp : \ok.\G,\\
&&&  \phantom{\HH \to \q : \{    }   \textit{nack}:   \q \to \pp : \textit{notyet}. \q \to \HH : \textit{transf}. \G_1,\\
&&& \phantom{\HH \to \q : \{    }  \textit{stop} : \q \to \pp :  \textit{stop}
\}
\end{array}$}}
 \end{example}

In order to assure lock-freedom by typing we require that the first occurrences of participants in global types are at a bounded depth in all paths starting from the root. This is formalised by the following definition of $weight$. 
\begin{definition}[Weight] \mylabel{df:w}
 Let $ \depth{\path_1\,  (\q \to \pr) \, \ell \, \path_2}\pp=length (\path_1)$ if $\pp \not \in \path_1$ and 
  $\pp \in \{ \q, \pr \}$, then\\ 
\centerline{$
 \depth \G\pp = \begin{cases}
max \{ \depth \path\pp \mid \path
  \in \paths \G  \}& \text{if }\pp\in\participantG\G, \\
0      & \text{otherwise}
\end{cases}
 $}
 \end{definition}
\begin{example}\mylabel{ex:w}
{\em
If $\G$ is as in Example \ref{ex:G}, then $ \depth \G\pp = \depth \G\q =0$,  and $\depth \G\HH =1$.
 If $\G'=\pp\to\q:\set{\ell_1.\pr \to \pp:\ell_3, \ell_2.\G'}$, then $ \depth {\G'}\pr = \infty$.
}
\end{example}

 The standard projection of  global types 
onto participants produces session types and session types are assigned to processes by a type system~\cite{HYC08,Honda2016}. The present simplified shape of messages
allows us  to define a projection of 
 global types onto participants producing processes instead of local types.

 The projection of a global  type onto a participant 
 returns, if any, the process  that the participant 
should run to follow the protocol specified by  the global type.  
If the global type 
begins by establishing a communication  from  $\pp$ to $\q$, then
the projection onto $\pp$ 
should send  one 
message  to $\q$, and 
the projection onto $\q$  should receive one 
message  from $\pp$.
The  projection  onto a third participant $\pr$  skips the initial communication,
that does not involve her. 
 This implies that the behaviour of $\pr$ must be
independent of the branch chosen by $\pp$, that is the projections on $\pr$ of
all the branches must be the same. However, in case of projections yielding input processes
from the same sender,
we can allow the process of $\pr$ to combine all these processes, proviso the messages are all distinct.

\begin{definition}[Projection] \mylabel{definition:projection}
Given a global type $\G$ and a participant $\pp$, 
we define  the 
partial function $\proj{ \ \ }{\pp}$
coinductively as follows:\\ 
\centerline{$
\begin{array}{rcll}
\proj{\G}\pp & =&\inact   \mbox{ \quad if  $\pp \not \in \participantG \G$} \\
\projp {\Gvti \pp\q \ell \S {\G}} \pp & =& \pout\q  \setlp{\ell} {\G}{n}{\pp}     \\
\projp {\Gvti \q \pp \ell \S {\G}} \pp & =& \pin\q \setlp{\ell}{\G}{n}{\pp}     &\\
\projp {\Gvti \q\pr \ell \S {\G}} \pp & =&  \begin{cases}
  \proj{\G_1}\pp  
    &  \text{if } \pp \not \in \{\q, \pr \}  \text{ and }\proj{\G_i}\pp=\proj{\G_j}\pp (1\leq i,j\leq n) \\
\pin \s (\mklab {\Lab_1}   {\mklab{\ldots}{\Lab_n}} ) &
 \text{if  }\pp \not \in \{\q, \pr \},  \ \  
\proj {\G_i} \pp = \pin \s  {\Lab_i}  \ (1\leq i\leq n) \text{ and }\\[-1.5mm]
& 
 \labels{\Lab_i}\cap\labels{\Lab_j}=\emptyset \text{ for } 1\leq i\not= j \leq n 
\end{cases}\\[20pt]
\end{array}
$}
We say that $\proj \G \pp$ is  the 
{\em projection}
of $\G$ onto $\pp$ if $\proj \G \pp$ is defined.
We say that $\G$ is {\em projectable} if $\proj \G \pp$ is  defined 
for all participants $\pp$. 
\end{definition}
This projection is the coinductive version of the projection given in \cite{DezEtAlt16,GJPSY19},  where processes are replaced by local types.  

 As mentioned above, if $\pp$ is not involved in the first communication of $\G$, and $\G$ starts with a choice between distinct messages, then in all branches the process of participant $\pp$ must either behave in the same way or be 
  a different  input, so that $\pp$ can understand which branch was chosen.

 \begin{example}
 \mylabel{ex:projG}
 {\em 
The global type $\G$ of Example \ref{ex:G}  is projectable, and by projecting it we obtain $\proj \G \pp=P$, $\proj \G \q=Q$, $\proj \G \HH=H$, 
where $P$, $Q$, and $H$ 
are as defined in Example \ref{ex:M}.  Also the global type $\G'$ of Example \ref{ex:w} is projectable, $\proj{\G'}\pp=\pout\q{\set{\ell_1.\pin\pr{\ell_3},\ell_2.\proj{\G'}\pp}}$, $\proj{\G'}\q=\pin\pp{\set{\ell_1, \ell_2.\proj{\G'}\q}}$, $\proj{\G'}\pr=\pout\pp{\ell_3}$.  Notice that $\G'$ has two branches, the projection of the first branch onto $\pr$ is $\pout\pp{\ell_3}$,  the projection of the second branch onto $\pr$ is just the projection of $\G'$ onto $\pr$, so $\pout\pp{\ell_3}$ is the (coinductive) projection of $\G'$ onto $\pr$. 
 }
 \end{example}
 
 \begin{definition}[Well-formed Global Types]
 A global type $\G$ is well formed if $\depth \G\pp$ is finite and $\proj\G\pp$ is defined for all $\pp\in \participantG\G$. 
 \end{definition}
  The global type $\G$ of Example \ref{ex:G}  is well formed, while  the global type $\G'$ of Example \ref{ex:w} is not well formed. 
In the following {\bf  we only consider well-formed global types}.

\smallskip
\noindent
 To type multiparty sessions we use  the  preorder $\subt$ 
on processes below, inspired by the subtyping of  \cite{CDG19}.  
\begin{definition}[Structural Preorder]
\mylabel{definition:subt}
We define the {\em structural preorder on processes}, $\PP\subt\Q$, by coinduction:\\
\centerline{$\inferrule[\rulename{sub-$\inact$}]{}
  {\inact \subt \inact}\qquad
\begin{array}{@{}l@{}}
\cinferrule[\rulename{sub-in}]{
 \PP_i \subt \Q_i \quad \forall 1 \leq i \leq n}{
  \pin\pp  (\mklab{\setl{\ell}{\PP}{n}}{\Lab}) \subt \pin \pp  \setl{\ell}{\Q}{n}
}
 \qquad 
\cinferrule[\rulename{sub-out}]{
  \PP_i \subt \Q_i \quad \forall 1 \leq i \leq n}{
  \pout\pp  \setl{\ell}{\PP}{n} \subt \pout \pp  {\setl{\ell}{\Q}{n}}
}
\end{array}
$}
\end{definition}
 The double line in rules indicates
that the rules are interpreted  {\em coinductively}. 
Rule \rulename{sub-in} allows bigger processes to offer fewer inputs, while
Rule \rulename{sub-out} requires the output messages  to be the same. 
The regularity condition on processes is crucial to guarantee the termination of 
   algorithms  for checking  structural  preorder.  
As  it will be further  discussed in Remark \ref{rem:counterexs},  
 the current proof fails if we type processes used as gateways by means of 
 the preorder  $\subtp$\mylabel{subtp}  obtained by substituting 
 rule\\ 
 \centerline{$\cinferrule[\rulename{sub-out$^+$}]{
  \PP_i \subt \Q_i \quad \forall 1 \leq i \leq n}{
  \pout\pp  \setl{\ell}{\PP}{n} \subt \pout \pp (\mklab{ \setl{\ell}{\Q}{n}}\Lab)
}$}
 for
 rule $\rulename{sub-out}$.  
 
\smallskip

The typing judgments 
  associate global types to sessions: they  are of the shape $\der{}\N\G$. 

\begin{definition}[Typing system] \mylabel{definition:typing} The only typing rule is:\\
\centerline{$\inferrule[\rulename{t-sess}]{\forall i\in I\quad\T_i\subt\proj\G{\pp_i}  
    \quad \participantG \G\subseteq\{\pp_i\mid i\in I\}}
  {\der{}{\prod\limits_{i\in I}\pa {\pp_i}\PP_i}\G
  }$} 
 \end{definition}
This rule requires that 
 the processes in parallel can play as participants of a whole communication protocol or  they are  the terminated process, i.e. they are smaller or equal (according to the structural preorder) to the projections of a unique global type. %
  The condition $\participantG \G\subseteq\{\pp_i\mid i\in I\}$ allows to type also sessions containing $\pa\pp\inact$, a property needed to assure invariance of types under structural congruence. 
 Notice that this typing rule allows
  to type multiparty session only with global types which can be projected on all their participants.  
  A session $\M$ is {\em well typed}  
    if there exists $\G$ such that
  $\vdash \M : \G$. 
  
  \mysection{Properties of Well-Typed Sessions}
 We start with the  standard  
 lemmas of inversion and canonical form, easily following  from Rule \rulename{t-sess}.

\begin{lemma}[Inversion Lemma]\mylabel{il}
If $\der{}{\prod\limits_{i\in I}\pa {\pp_i}\PP_i}\G$, then $\T_i\subt\proj\G{\pp_i}$  for all $i\in I$  and  
$ \participantG \G\subseteq\{\pp_i\mid i\in I\}$.
\end{lemma}

\begin{lemma}[Canonical Form Lemma]\mylabel{cfl}\hfill\\
If $\der{}{\N}\G$ and $ \participantG \G=\{\pp_i\mid i\in I\}$, then $\N\equiv\prod\limits_{i\in I}\pa {\pp_i}\PP_i$ and $\T_i\subt\proj\G{\pp_i}$ for all $i\in I$. 
\end{lemma}

 To formalise the properties of Subject Reduction and Session
Fidelity \cite{HYC08,Honda2016},  we use the standard LTS for global types
given  below. 
Rule \rulename{Icomm} is justified by the fact that in
a projectable global type $\mkt \pr \ps \Labgt $, the behaviours of a
participant $\pp$ different from $\pr$
and $\ps$ and starting with an output are the same in all branches, and hence they are
independent from the choice of $\pr$, and may be executed before it.

\begin{definition}[LTS for Global Types]
The {\em   labelled transition   system   (LTS) 
for global types}  
  is specified by the rules:\\
 \centerline{$
 \inferrule[\rulename{ecomm}]{}{ \mkt \pp \q (\mklab{\ell.\G} \Labgt) \redparr{\pp \ell \pq}\G}
 \qquad\quad
 \inferrule[\rulename{icomm}]{  \G_i \redparr{\pp \ell \pq}\G'_i \quad\set{\pp,\q}\cap\set{\pr,\ps}=\emptyset\quad\text{\em for all }i~(1\leq i\leq n)
     }{
   \Gvti \pr \ps \ell \S {\G}
    \redparr{\pp \ell \pq}
   \Gvti \pr \ps \ell \S {\G'}
    }
    $}
    \end{definition}
    
    The following lemma relates projections and reductions of global types.
    \begin{lemma}[Key Lemma]\mylabel{kl}
 \begin{enumerate}[label=(\roman*)] \item\mylabel{kl1} If $\proj\G\pp=\pout\q\Lab$ and $\proj\G\q=\pin\pp{\Lab'}$, then $\labels\Lab=\labels{\Lab'}$.  Moreover  $\G\redparr{\pp \ell \pq}\G^\ell$ and $\ell.\proj{\G^\ell}\pp\in\Lab$ and $\ell.\proj{\G^\ell}\q\in\Lab'$ for all $\ell\in\labels\Lab$. 
\item\mylabel{kl2} If $\G\redparr{\pp \ell \pq}\G'$, then $\proj\G\pp=\pout\q\Lab$ and   $\proj\G\q=\pin\pp{\Lab'}$ and $\ell\in\labels\Lab=\labels{\Lab'}$. 
\end{enumerate}
\end{lemma}
\begin{myproof}  
\ref{kl1}. The proof is by  induction on  $n= \depth{\G}{\pp}$.
If $n=0$,  then   we have  $\G=\mkt \pp \q \Labgt$ and $\labels\Labgt=\labels\Lab=\labels{\Lab'}$ 
 and $\ell.\proj{\G^\ell}\pp\in\Lab$ and $\ell.\proj{\G^\ell}\q\in\Lab'$   by definition of projection.
If $n >0$, then $\G=\Gvti \pr \ps \ell \S {\G}$ and  $\{\pp, \q \} \cap   \{\pr, \s \} = \emptyset$ and $\proj{\G_i}\pp=\pout\q\Lab$ and $\proj{\G_i}\q=\pin\pp{\Lab'}$ for all $i$, $1\leq i\leq n$, by definition of projection.  By  the  induction hypothesis, $\labels\Lab=\labels{\Lab'}$.  Moreover, again by  the  induction  hypothesis,  $\G_i\redparr{\pp \ell \pq}\G_i^\ell$ and   $\ell.\proj{\G_i^\ell}\pp\in\Lab$ and $\ell.\proj{\G_i^\ell}\q\in\Lab'$ for all $i$, $1\leq i\leq n$. We get $\G\redparr{\pp \ell \pq}\G^\ell$  using rule $\rulename{icomm}$, where $\G^\ell=\Gvti \pr \ps \ell \S {\G^\ell}$.  The definition of projection implies   $\proj{\G^\ell}\pp=\proj{\G_1^\ell}\pp$ and $\proj{\G^\ell}\q=\proj{\G_1^\ell}\q$.\\  
 \ref{kl2}. The proof is by  induction on  $\depth{\G}{\pp}$ and by cases on the reduction rules.\\
 The case of  rule $\rulename{ecomm}$ is easy. For  rule    $ \rulename{icomm}$, 
 by the induction hypothesis, $\proj{\G_i}\pp=\pout\q{\Lab_i}$ and  $\proj{\G_i}\q=\pin\pp{\Lab'_i}$ and $\ell\in\labels{\Lab_i}=\labels{\Lab'_i}$  for all $i$, $1\leq i\leq n$. By definition of projection $\proj{\G_i}\pp=\proj{\G_j}\pp$ and  $\proj{\G_i}\pq=\proj{\G_j}\pq$  for $1\leq i,j\leq n$.  Again by definition of projection  $\proj{\G}\pp=\proj{\G_1}\pp$ and $\proj{\G}\q=\proj{\G_1}\q$.  \QED
\end{myproof}

\medskip

Subject Reduction says that the transitions of well-typed sessions are mimicked by those of global types. 
\begin{theorem}[Subject Reduction]
\mylabel{thm:SR}
If  $\der{}\M\G$ 
and  $\M \redparr{\pp \ell \q} \M'$, then $\G\redparr{\pp \ell \q}\G'$
and 
$\der{} {\M'} {\G'}$.
\end{theorem}
\begin{myproof} 
  If $\M \redparr{\pp \ell \q} \M'$, then\\ 
 \centerline{$
\begin{array}{lll}
\M  & \equiv  & 
 \pa{\pp}  \pout{\q}{ (\mklab{\ell.\PP}{\Lab}) }
    \; \pc \;\pa{\q}{\pin {\pp} { (\mklab{\ell.\Q}{\Lab'} }})  \; \pc \; 
    \prod\limits_{1\leq j\leq m} \pa{\pr_j}{\R_j}
    \\[-1ex]
\M' & \equiv &
    \pa{\pp}{\PP}\;\pc\;\pa{\q}{\Q} \; \pc \;  \prod\limits_{1\leq j\leq m} \pa{\pr_j}{\R_j} 
\end{array}
$}
 Since $\der{}\M\G$,  by Lemma~\ref{il}  
 we have that 
 $\pout{\q}{ (\mklab{\ell.\PP}{\Lab})}  \subt  \proj \G {\pp}$,
$\pin {\pp} { (\mklab{\ell.\Q}{\Lab'}) } \subt \proj \G {\q}$
and
$\R_j \subt \proj \G {\pr_j}$  $(1\leq j\leq m)$. 
By definition of $\subt$, from $\pout{\q}{ (\mklab{\ell.\PP}{\Lab})}  \subt  \proj \G {\pp}$
we get  $\proj \G {\pp}=\pout{\q} {(\mklab{\ell.\PP_0}{\Lab_0})}$  and $\PP\subt\PP_0$.
Similarly from $\pin {\pp} { (\mklab{\ell.\Q}{\Lab'}) } \subt \proj \G {\q}$ we get
$\proj \G {\q}=\pin {\pp} { (\mklab{\ell.\Q_0}{\Lab'_0}) }$ and $\Q\subt\Q_0$.
 Lemma~\ref{kl}\ref{kl1}  implies $\G\redparr{\pp \ell \q}\G'$ and $\proj{\G'}\pp=\PP_0$ and $\proj{\G'}\q=\Q_0$. 
  We show  $\proj \G {\pr_j}\subt\proj {\G'} {\pr_j}$ for each $j$, $1\leq j\leq m$ by induction on $\depth\G{\pr_j}$ and by cases on the reduction rules.  For rule $\rulename{ecomm}$ we get $\G=\pp\to\q:(\mklab{\ell.\G'} \Labgt)$. By Definition  \ref{definition:projection} either 
 $\proj \G {\pr_j}=\proj {\G'} {\pr_j}$ or $\proj \G {\pr_j}\subt\proj {\G'} {\pr_j}$. For rule $\rulename{icomm}$  $\proj {\G_i} {\pr_j}\subt\proj {\G'_i} {\pr_j}$ for $1\leq i\leq n $ by  the  induction hypothesis. In both cases $\proj \G {\pr_j}\subt\proj {\G'} {\pr_j}$. 
  We conclude $\der{} {\M'} {\G'}$.  \QED
 \end{myproof}
 
 \smallskip
 
 Session fidelity assures that the communications in a session typed by a global type are done as prescribed by the global  type. 
 \begin{theorem}[Session Fidelity]\mylabel{sft}
 Let $ \der{}\M\G$. 
\begin{enumerate}[label=(\roman*)]
\item \mylabel{sft1}If $\M \redparr{\pp \ell \q} \M'$, then $\G\redparr{\pp \ell \q}\G'$ and $\der{} {\M'} {\G'}$.  
\item \mylabel{sft2} If 
$\G \redparr{\pp \ell \q} \G'$, then  $\M\redparr{\pp \ell \q}\M'$   and $\der{} {\M'} {\G'}$. 
\end{enumerate}
\end{theorem}
\begin{myproof}
\ref{sft1}. It is the Subject Reduction Theorem.\\
\ref{sft2}. By Lemma~\ref{kl}\ref{kl2},  
$\proj\G\pp=\pout\q\Lab$  and   $\proj\G\q=\pin\pp{\Lab'}$ and $\ell\in\labels\Lab=\labels{\Lab'}$. 
  By Lemma~\ref{kl}\ref{kl1}, $\ell.\proj{\G'}\pp\in\Lab$ and $\ell.\proj{\G'}\q\in\Lab'$. 
By Lemma~\ref{cfl},
$\M\equiv\pa\pp{\PP}\; \pc \;\pa\q{\Q}\; \pc \;\M_0$ and 
$\PP\subt\proj\G\pp$ and $\Q\subt\proj\G\q$. 
 By definition of $\subt$ we get   $\PP=\pout{\q}{(\mklab{\ell.\PP'}{ \Lab_1})}$ with   $\labels{\Lab} = \set\ell\cup\labels{\Lab_1}$  and  $\PP'\subt\proj{\G'}\pp$,  
 and $\Q=\pin {\pp} { (\mklab{\ell.\Q'}{ \Lab_2}) }$ with 
  $\labels{\Lab_2}\cup \set\ell \supseteq \labels{\Lab'}$ and 
 $\Q'\subt\proj{\G'}\q$. 
 Hence {$\M \redparr{\pp \ell  \q  }
\pa\pp{\PP'} \pc \pa\q{\Q'} \pc \M_0\ =\ \M'$} and $\der{} {\M'} {\G'}$.  \QED
  \end{myproof} 
 
 \medskip
 
 Let $\vdash^+$ be the typing system obtained by using $\subtp$  (as defined on page \pageref{subtp})  in rule $\rulename{t-sess}$. 
  Session Fidelity for $\vdash^+$ is weaker than for $\vdash$.  If $\N=\pa\pp\pout\q.\ell_1\pc\pa\q\pin\pp\set{\ell_1,\ell_2}$ and  $\G=\pp\to\q:\set{\ell_1,\ell_2}$, then $\derp{}{\N}{\G}$ and  $\G\redparr{\pp \ell_i \q}\tend$ with $i=1,2$, but the only reduction of $\N$ is $\N\redparr{\pp \ell_1 \q}\pa\pp\inact$. Notice that  $\pout\q.\ell_1 \subtp
\proj\G\pp$ but  $\pout\q.\ell_1\not\subt
\proj\G\pp$ and $\pin\pp\set{\ell_1,\ell_2}=\proj\G\q$.\\
 Clearly $\der{}\N{\G}$ implies $\derp{}\N\G$,  and  a weakening of  the vice versa is shown below.  
\begin{theorem}
If $\derp{}\N\G$, then $\der{}\N{\G'}$ for some $\G'$.
\end{theorem}
\begin{proof} The proof is by coinduction on $\G$. Let $\G=\pp\to\q:\Labgt$.  Then,  
by Lemmas \ref{cfl} and \ref{il} (which easily extend to $\vdash^+$),  Definition \ref{definition:projection} and the definition of $\subt^+$, $\N\equiv \pa\pp\pout\q\Lab\pc\pa\q\pin\pp{\Lab'}\pc\N'$ and $\pout\q\Lab\subtp\proj\G\pp$ and $\pin\pp{\Lab'}\subtp\proj\G\q$.  
Again by the definition of $\subt^+$, $\labels{\Lab}\subseteq\labels{\Labgt}\subseteq\labels{\Lab'}$. Let 
$\Lab=\setl{\ell}{\PP}{n}$, $\Labgt=\mklab {\setl{\ell}{\G}{n}}{\Labgt'}$ and $\Lab'=\mklab {\setl{\ell}{\Q}{n}}{\Lab''}$.
Then $\N\redparr{\pp\ell_i\q}\pa\pp{\PP_i}\pc\pa\q{\Q_i}\pc\N'$ and $\G\redparr{\pp\ell_i\q}\G_i$ for all $i$, $1\leq i\leq n$. Since the proof 
of Theorem \ref{thm:SR} easily adapts to $\vdash^+$, we get 
$\derp{}{\pa\pp{\PP_i}\pc\pa\q{\Q_i}\pc\N'}{\G_i}$ for all $i$, $1\!\leq i\!\leq n$. 
By coinduction there are $\G'_i$ such that $\der{}{\pa\pp{\PP_i}\pc\pa\q{\Q_i}\pc\N'}{\G'_i}$ for all $i$, $1\!\leq i\!\leq n$. We can choose $\G'=\pp\to\q:\setl{\ell}{\G'}{n}$. \QED
\end{proof} 

\medskip
  
   We end this  section by showing that the type system $\vdash$ assures lock-freedom. 
 By Subject Reduction it is enough to  prove 
that well-typed  sessions are deadlock-free and no participant waits forever.
The former follows from  Session Fidelity,
while the latter follows from the following lemma that says that reducing by rule $\rulename{ecomm}$ the weights of the  not involved  participants   
strictly decrease.  

 \begin{lemma}\mylabel{l:e}
 If $\G \redparr{\pp \ell \q} \G'$ by rule \text{\em $\rulename{ecomm}$} and $\pr\not\in\set{\pp,\q}$ and $\pr\in\participantG\G$, then $\depth\G\pr>\depth{\G'}\pr$.  
 \end{lemma}
 \begin{myproof}
 Rule $\rulename{ecomm}$ implies $\G= \pp \to \q : (\mklab{\ell.\G'} \Labgt)$. If $ \path$ is a path of $\G'$, then  $(\pp \to \q) \, \ell \, \path$ is a path in $\G$. This gives $\depth\G\pr>\depth{\G'}\pr$. \QED
   \end{myproof} 
   
  Multiparty session typability guarantees 
  lock-freedom.
\begin{theorem}[Lock-Freedom]
\mylabel{theorem:waitfreedom}   
If session $\M$  is    well typed,  then
$\M$ is  lock-free. 
\end{theorem}
 \begin{myproof}
 Let $\G$ be a type for $\M$. If $\M\not\equiv\pa\pp\inact$, then $\G\not=\tend$. Let $\G=\mkt \q \pr \Labgt$. By definition of reduction $\G\redparr{\q\ell\pr}\G'$ for some $\ell$, and this implies  $\M\redparr{\q\ell\pr}\M'$  by Theorem \ref{sft}\ref{sft2}. This shows condition \ref{d:lf1} of Definition \ref{d:lf}.
 The proof of condition \ref{d:lf2} of Definition \ref{d:lf} is by induction on $n=\depth\G\pp$. If $n=0$ then either $\G=\mkt \pp \q \Labgt$ or $\G=\mkt \q \pp \Labgt$ and $\G\redparr{\lambda} \G'$ with $\pp$ in $\lambda$  by rule $\rulename{ecomm}$. If $n>0$ then $\G=\mkt \q \pr \Labgt$ with $\pp\not\in\set{\q,\pr}$ and $\G \redparr{\q \ell \pr} \G'$  for all $\ell\in\labels\Labgt$  by rule $\rulename{ecomm}$. By Lemma \ref{l:e} $\depth\G\pp>\depth{\G'}\pp$  and induction applies.  \QED
  \end{myproof} 
  
  It is easy to check that $\der{}\M\G$, where $\M$ and $\G$ are the multiparty session and the global type
  of  Examples \ref{ex:M} and  \ref{ex:G}, respectively.  By the above result, $\M$  of Example \ref{ex:M}  is hence provably lock-free.

%% file: mergeN.tex

\mysection{Connection of  Multiparty-Sessions  via Gateways}

Given two multiparty sessions, they can be {\em connected via gateways} when
they possess two {\em compatible} participants, i.e. participants that offer communications which can be paired
and can hence be transformed into forwarders, that we dub ``gateways''.  
 We  start by discussing 
the relation of compatibility between processes 
by elaborating on  Examples \ref{ex:M} and  \ref{ex:G}.\\
If we decide to look at the participant $\HH$ as an interface,
the  messages sent by her have to be considered as those actually provided by an external environment;
 and  the received messages  as messages expected by such an environment.
 In a sense, this means that, if we abstract from participants' names in the process $H$, we get a description
 of an interface (in the more usual sense) of   
 an external system, rather than an interface of our system.
 
 In order to better grasp the notion of compatibility hinted at above, let us 
 dub  ``$\eraseP$''  the operation abtracting from the participants' names inside processes.
 So, in our example we would get \\
 \centerline{$
\begin{array}{lll}
\eraseP(H) = \pin {\circ}{\textit{text}.  \eraseP(H_1)  $ \qquad $ \eraseP(H_1) = \pout{\circ}{\{ \textit{ack}.\eraseP(H), \textit{nack}.{\pin {\circ}{\textit{transf}. \eraseP(H_1) }},  \textit{stop}  
\}}}
\end{array}
$}
Let us now take into account another system that could work as the environment of the system having the
$\G$ of Example \ref{ex:G} as global  type. 
Let assume such a system to be formed by
participants $\KK$, $\pr$ and $\ps$ interacting according the following protocol:
\begin{description}
 \item  Participant  
$\KK$ sends text messages to $\pr$ and $\ps$ in an alternating way,  starting with $\pr$.  
\item Participants  
$\pr$ and $\ps$ inform $\KK$ that a text has been accepted or refused by sending back, respectively,\\[-0.5mm] 
either 
$\textit{ack}$ or $\textit{nack}$. 
\begin{description}
\item 
In the first case it is the other receiver's turn  to receive the  text: a message \textit{go} is exchanged between $\pr$ and $\ps$ to signal this case; 
\item
in the second case, the sender has to resend the text until it is accepted. Meanwhile  the 
involved participant between $\pr$ and $\ps$ informs the other one that she needs to {\em wait} since the previous message is being resent in a \textit{transf}ormed form.
\end{description}
\end{description}
 This protocol  can be  implemented by the multiparty session\mylabel{smp}  \qquad
$
\M'= \pa {\pr}  R   \pc \pa {\ps} S  \pc  \pa \KK  K_{ \pr } 
$\\\
 where\\
 \centerline{$\begin{array}{lcl@{\quad}lcl@{\quad}lcl}
  R  & =&\pin {\KK}{\textit{text}.R_1} &R_1&=&\pout {\KK}{\set{\textit{ack}. \pout\ps \textit{go}. R_2, \textit{nack}.\ps!\textit{wait}.\pin{\KK}{\textit{transf}.R_1} }} &R_2 &=& \pin{\ps}{\set{\textit{go}.R,\textit{wait}.R_2}}\\
   \end{array}$}
  $\begin{array}{@{\quad}lcl@{\qquad}lcl@{\qquad}lcl}
   S&=&\pin\pr\set{\textit{go}.\pin {\KK}{\textit{text}.S_1, \textit{wait}.S}}&S_1&=&\pout {\KK}{\set{\textit{ack}. \pout\pr \textit{go}. S, \textit{nack}.\pr!\textit{wait}. \pin{\KK}{\textit{transf}.S_1} }}\\ 
   K_\pr &=& \pout {\pr}{\textit{text}.  K_\pr'   &  K_\pr'  &=&  \pin{\pr}{\{ \textit{ack}. K_\ps  }, \textit{nack}.\pout { \pr  }{\textit{transf}. K_\pr'   \}}}\\
 K_\ps   &=& \pout {\ps}{\textit{text}.  K'_\ps  }
&
 K'_\ps  &=& \pin{\ps}{\{ \textit{ack}.K_\pr, \textit{nack}.\pout { \ps  }{\textit{transf}.\text{$ K'_\ps $} }\}}
 \end{array}$\\

 The ``behaviour as interface'' of participant $\KK$  corresponds to\\
\centerline{$
\eraseP(K_\pr)  = \eraseP(K_\ps)  =   \pout {\circ}{\textit{text}. \eraseP(K_\pr')    \quad   \eraseP(K_\pr') = \eraseP(K_\ps') =   \pin{\circ}{\{ \textit{ack.}\text{$\eraseP(K_\pr)$}, \textit{nack.}\circ ! \textit{transf}.\text{$\eraseP( K_\pr')$}  \}}}
$}
Notice that the mapping $\eraseP$ equates $K_\pr$ and $ K_\ps $, i.e. $\eraseP(K_\pr)  = \eraseP(K_\ps)$.\\
 The  interactions ``offered'' and ``requested'' by $\eraseP(H)$ and $\eraseP(K_\pr)$ do not
precisely match each other, that is $\Dual{\eraseP(H)} \neq \eraseP(K_\pr)$ (where $\Dual{(\cdot)}$ is the standard syntactic duality function replacing `!' by `?' and vice versa  \cite{HVK98}).  
Nonetheless it is easy to check that, even if  the system $\pa\pp\PP\pc\pa\q\Q$ of Example \ref{ex:M}     can safely deal with a message $\textit{stop}$
coming from its environment, no problem  arises  in case no such a message will ever arrive.\\
 In the following definition, instead of explicitly introduce the ``$\eraseP$'' function,  we simply
 formalise  the compatibility relation in such a way
two processes are compatible (as interfaces) whenever they offer dual communications to {\em arbitrary} participants, and the set of input labels is a subset of the set of output labels. 

\begin{definition}[Processes'  Compatibility]
\mylabel{def:compatibility}
The {\em interface compatibility relation} $\PP\interfacecomp\Q$ on processes ({\em compatibility} for short), 
 is the largest symmetric relation coinductively defined by:\\
\centerline{$
\begin{array}{c@{\qquad\qquad\qquad\qquad}c}
\inferrule[\rulename{comp-$\inact$}]{}
  {\inact \interfacecomp \inact}
&
\cinferrule[\rulename{comp-{ \sc o/i  
}}]{
 \PP_i \interfacecomp \Q_i \quad \forall 1 \leq i \leq n}{
  \pout\pp  (\mklab{\setl{\ell}{\PP}{n}}{\Lab})    \interfacecomp    \pin \pq  \setl{\ell}{\Q}{n}
}
\end{array}
$}
\end{definition}
 The double line in rule  $\rulename{comp-{\sc o/i }}$  indicates
that the rule is  {\em coinductive}.  
 Notice that the relation $\interfacecomp$ is insensitive to the names of senders and receivers. It is immediate to verify that process compatibility is similar and simpler than 
  the  subtyping defined in \cite{GH05}. Therefore an algorithm for checking process compatibility can be an easy adaptation of the  algorithm given in
 \cite{GH05}.

\noindent
For what concerns our example, it is  straightforward 
to verify that $ H \interfacecomp  K_\pr $.

 Useful properties of compatibility are stated in the following proposition, whose proof is simple. 

\begin{proposition}\mylabel{p:gw}
\begin{enumerate}[label=(\roman*)]
\item\mylabel{p:gw1} If  $\PP    \interfacecomp    \pin \pp  (\mklab\Lab{\Lab'})$, then $\PP    \interfacecomp    \pin \pp \Lab$.
\item\mylabel{p:gw2} If $\pout \pp (\mklab{\ell.P}\Lab)  \interfacecomp \pin\q{\ell.\Q}$, then $\PP \interfacecomp\Q$.
\end{enumerate}
\end{proposition}

Similarly to what is done in \cite{BdLH18}  for 
the setting of  Communicating Finite State Machines (CFSMs),  the presence of  two compatible processes 
$H$ and $K$
in  two  multiparty sessions $\M$ and $\M'$ enables to connect these systems by transforming  
$H$ and $K$ in such a way  
each message received by $H$ is immediately sent to $K$,
 and each message sent by $H$ is first received from $K$.  And similarly for  what concerns $K$.
In the following definition we hence 
 transform  an arbitrary process $\PP$ not containing a fixed participant $\HH$ into a process which:
1.  sends to $\HH$ each message received in $\PP$;\quad
2.  receives from $\HH$ each message sent in $\PP$.\\
We call $\gateway{P,\HH}$ the so obtained process.

\begin{definition}[Gateway Process]\mylabel{def:gc}
Let $\HH\not\in\participantP \PP$.
We define 
$\gateway{P,\HH}$ coinductively as follows\\
\centerline{$
\begin{array}{lll}
\gateway{\inact,\HH} & = & \inact
\\
\gateway{\pin{\pp}{\Set{\ell_i.P_i \mid {\scriptstyle 1 \leq i \leq n }}},\HH} & = &
         \pin{\pp}{\Set{\ell_i.\pout{\HH}{\ell_i.\gateway{P_i,\HH}} \mid {\scriptstyle 1 \leq i \leq n }}}
\\
\gateway{\pout{\pp}{\Set{\ell_i.P_i \mid {\scriptstyle 1 \leq i \leq n }}},\HH} & = &
         \pin{\HH}{\Set{\ell_i.\pout{\pp}{\ell_i.\gateway{P_i,\HH}} \mid {\scriptstyle 1 \leq i \leq n }}}
\end{array}
$}
\end{definition}

 A first lemma assures the soundness of  the  previous definition. 

\begin{lemma} If $\HH\not\in\participantP \PP$, then $\gateway{\PP,\HH}$ is defined and is a function. 
\end{lemma}
\begin{myproof}  The proof is  
by coinduction. If $\PP=\pin{\pp}{\Set{\ell_i.\PP_i \mid {\scriptstyle 1 \leq i \leq n }}}$, then\\
\centerline{$\gateway{\pin{\pp}{\Set{\ell_i.\PP_i \mid {\scriptstyle 1 \leq i \leq n }}},\HH}  = 
         \pin{\pp}{\Set{\ell_i.\pout{\HH}{\ell_i.\gateway{P_i,\HH}} \mid {\scriptstyle 1 \leq i \leq n }}}$} By coinduction  $\gateway{P_i,\HH}$ is defined and is a function for $1 \leq i \leq n$.  The thesis hence follows.  Similarly when $P$ is an output process. \QED
\end{myproof}

The gateway process construction enjoys the 
preservation of 
the structural preorder.
This property is  the key to get  
Theorem \ref{t:m} below and it  essentially relies on the fact that bigger processes  offer the same  output messages. 
 \begin{lemma}\mylabel{l:gp} Let $\HH\not\in\participantP \PP\cup\participantP \Q$.
 If $P\subt Q$, then  
$\gateway{P,\HH} \subt \gateway{Q,\HH}$.
\end{lemma}
\begin{myproof}
We only consider the case of input processes, the proof for output processes is similar  and simpler.\\
If $P=\pin\pp  \setl{\ell}{P}{n}$ and  $Q= \pin \pp  \setl{\ell}{\Q}{n'}$ with $n'\!\leq\! n$,
then $\gateway{P,\HH}=\pin{\pp}{\Set{\ell_i.\pout{\HH}{\ell_i.\gateway{P_i,\HH}} \mid {\scriptstyle 1 \leq i \leq n }}}$ 
and $\gateway{Q,\HH}=\pin{\pp}{\Set{\ell_i.\pout{\HH}{\ell_i.\gateway{Q_i,\HH}} \mid {\scriptstyle 1 \leq i \leq n' }}}$. From $P\subt Q$ we get $P_i\subt Q_i$ for all $i$, $1\leq i\leq n'$. By coinduction $\gateway{P_i,\HH} \subt \gateway{Q_i,\HH}$, which implies  $\pout{\HH}{\ell_i.\gateway{P_i,\HH}}\subt\pout{\HH}{\ell_i.\gateway{Q_i,\HH}}$ for all $i$, $1\leq i\leq n'$,  and hence 
$\gateway{P,\HH} \subt \gateway{Q,\HH}$, by definition of $\subt$ (Definition \ref{definition:subt}).  
  \QED
\end{myproof}

  Lemma \ref{l:gp} fails for $\subtp$. For example,  if $\PP=\pout\pp{\ell_1}$ and $\Q=\pout\pp{\set{\ell_1,\ell_2}}$, then $\PP\subtp\Q$, but $\gateway{\PP,\HH}=\pin\HH{\ell_1}.\pout\pp{\ell_1}\suptp\pin\HH\set{\ell_1.\pout\pp{\ell_1}, \ell_2.\pout\pp{\ell_1}}=\gateway{\Q,\HH}$.

The following relationship between compatibility and structural preorder of processes will be essential in the proof of 
our main result (Theorem \ref{t:m}).
\begin{lemma}
\mylabel{lem:compsubt} If 
$P\interfacecomp Q$, then  
$\PP\subt \PP'$ and  $Q\subt Q'$ imply $\PP'\interfacecomp Q'$.
\end{lemma}
\begin{myproof}
 Let us assume\\[-4mm]
\centerline{
$\qquad\qquad\qquad\quad\begin{array}{llcl@{\qquad}l}
 \text{} P=\pout\pp  {\setl{\ell}{\PP}{ n}} &  \subt & P'=\pout\pp  {\setl{\ell}{\PP'}{n}}
 \\[-1mm]
  \updownarrow\\
  Q=\pin \pq  \setl{\ell}{\Q}{n'}  &  \subt & Q'=\pin \pq  \setl{\ell}{\Q'}{n''} & \text{with } n''\leq n'\leq n  
  \end{array}$}
   From $P\interfacecomp Q$ we get $\PP_i \interfacecomp \Q_i$ for all $i$, $1 \leq i \leq n'$.  From $\PP\subt \PP'$ we get 
   $\PP_i\subt \PP_i'$ for all $i$, $1 \leq i \leq n$. 
   From $\Q\subt \Q'$ we get 
 $\Q_i\subt \Q_i'$ for all $i$, $1 \leq i \leq n''$. By coinduction we have $\PP_i '\interfacecomp \Q_i'$ for all $i$, $1 \leq i \leq n''$. We can then conclude $\PP'\interfacecomp Q'$. \QED
\end{myproof}

\noindent
The vice versa does not hold. For example $\pout\pp\ell\interfacecomp\pin\q\ell$  and $\pin\q{\set{\ell,\ell'}}\subt\pin\q\ell$,  but $\pout\pp\ell\interfacecomp\pin\q{\set{\ell,\ell'}}$ is false.

\medskip

The formal definition of connection of multiparty sessions via gateways  is 
based on the  
notion of process compatibility (Definition \ref{def:compatibility}) and on the 
addition of communications to a process (Definition \ref{def:gc}). 

\begin{definition}[Multiparty-Sessions' Compatibility]\mylabel{msc}\hfill\\
Two multiparty sessions $\M$, $\M'$ are \em{compatible} via the participants $\HH$, $\KK$ (notation $\pair\M\HH\interfacecomp\pair{\M'}\KK)$ if\\
\centerline{ $\participantS{\M}\cap\participantS{\M'}=\emptyset\quad$ and $\quad\M\equiv\M_1\pc\pa\HH H\quad$  and   $\quad \M'\equiv\M_2\pc\pa\KK K\quad$ with $H\interfacecomp K$.}
\end{definition}

\begin{definition}[Multiparty-Sessions' Connection via Gateways]\hfill\\
\mylabel{def:mptscomp}
Let $\M\equiv\M_1\pc\pa\HH H$, $\M'\equiv\M_2\pc\pa\KK K$ and $\pair\M\HH\interfacecomp\pair{\M'}\KK$.
We define $\M\conn{\HH}{\KK}\!\M'$, the  connection   of $\M$ and $\M'$ via gateways, through $\HH$ and $\KK$,
by\\
\centerline{$\M\conn{\HH}{\KK}\!\M' \ByDef 
\M_1 \pc \M_2
\pc    \pa{\HH}{\gateway{H,\KK}}   \pc    \pa{\KK}{\gateway{K,\HH}}$}
\end{definition}
\begin{example}\mylabel{ex:msc} 
{\em
For what concerns  $\M$ of Example \ref{ex:M}, $\M'$ defined on page \pageref{smp},  $\HH$ and $\KK$,  
it is not difficult to check that\\
\centerline{$
\M\conn{\HH}{\KK}\!\M' \quad=\quad \pa \pp  \PP    \pc   \pa \q \Q       
   \pc   \pa {\pr}  R   \pc \pa {\ps} S    \pc      \pa \HH \hat{H}    \pc  \pa \KK \hat{K_\pr}   
$}\\
 where\\[3pt]
$\begin{array}{@{\hspace{-0pt}}c}
\hat{H} \!=\! \gateway{H,\KK} \!=\! \pin {\q}{\textit{text}. \pout{\KK}{\textit{text}}.\hat{H_1}\quad\hat{H_1}\!=\! \pin{\KK}{\{ \textit{ack}.\pout{\q}{\textit{ack}}.\hat{H},\  \textit{nack}.\pout{\q}{\textit{nack}}.\pin\q{\textit{transf}}.\pout\KK{\textit{transf}}.\hat{H_1},\  \textit{stop}.\pout{\q}{\textit{stop}}  \}}}\\
 \hat{K_\pr} = \gateway{K_\pr,\HH} = \pin {\HH}{\textit{text}}. \pout {\pr}{\textit{text}. \hat{K'_\pr} \quad  \hat{K'_\pr}=\pin{\pr}{\{ \textit{ack}.\pout{\HH}{\textit{ack}}.\hat{K_\ps},\  \textit{nack}.\pout{\HH}\textit{nack}.\pin\HH{\textit{transf}}.\pout {\pr}{\textit{transf}}.\hat{K_\pr'}}  \}}\\
 \hat{K_\ps} = \gateway{K_\ps,\HH} = \pin {\HH}{\textit{text}}. \pout {\ps}{\textit{text}. \hat{K'_\ps} \quad  \hat{K'_\ps}=\pin{\ps}{\{ \textit{ack}.\pout{\HH}{\textit{ack}}.\hat{K_\pr},\  \textit{nack}.\pout{\HH}\textit{nack}.\pin\HH{\textit{transf}}.\pout\ps{\textit{transf}}.\hat{K_\ps'}}  \}}
\end{array}
$}
\end{example}

 In the following section we shall prove that lock-freedom is preserved by the session connection via gateways.
This follows from the fact that we define an operator building a global type 
such that the participant processes of the 
session obtained by connection via gateways are smaller than or equal to the projections of this global type.

\mysection{Connection of Global Types via Gateways}

The composition defined in the previous section can be shown to be 
lock-freedom preserving
by means of Theorem  \ref{theorem:waitfreedom}. In fact it is possible to define a function
on global types with compatible participants, which 
 corresponds   to the lifting
of the construction in 
Definition \ref{def:mptscomp} to the level of global types.

\begin{definition}[Global-Types' Compatibility]
Two global types $\G$, $\G'$ are \em{compatible} via the participants $\HH$, $\KK$ (notation $\pair\G\HH\interfacecomp\pair{\G'}\KK)$ if $\participantG{\G}\cap\participantG{\G'}=\emptyset$ and $\proj{\G}{\HH}\interfacecomp\proj{\G'}{\KK}$.
\end{definition}

\begin{definition}[Global-Types' Connection via Gateways] \mylabel{d:gc}
Let $\pair\G\HH\interfacecomp\pair{\G'}\KK$. We define\\
\centerline{$\G\conn{\HH}{\KK}\G' \ByDef   \mergeG({\HH},{\KK},\noint,\G, \G')$}
where $\mergeG$ is coinductively  given 
by the following 
clauses,
assuming  
$\set{\pp,\q,\pr,\ps}\cap\set{\HH,\KK}=\emptyset$.\\
The clauses must be applied in the given order.\\
\centerline{$\begin{array}{@{\hspace{-2pt}}llcl}
\mbox{\scriptsize $(1)$} & \mergeG({\HH},{\KK},\noint,\tend, \G')  &=& \G'\\[1.5mm]            
\mbox{\scriptsize $(2)$}   &    \mergeG({\HH},{\KK},\noint,{\Gvti \pp\HH \ell \S {\G}}, \G') &=&
      \pp\to\HH : \Set{\ell_i. {\mergeG({\HH},{\KK},\ell_i^\to,\G_i,\G') } \mid  {\scriptstyle 1 \leq i \leq n }}\\[1.5mm]
\mbox{\scriptsize $(3)$}   &  \mergeG({\HH},{\KK},\ell^\to,\G, \KK\to\ps: \Set{\ell'_j.\G'_j\mid {\scriptstyle 1 \leq j \leq m }})  &=& \HH\to\KK: \ell. \KK\to\ps : \ell. \mergeG({\HH},{\KK},\noint,\G,\G'_\vr)  \\
& & &
 \textit{if $\ell= \ell'_\vr$ with $1 \leq \vr \leq m$}
 \\[1.5mm]                                   
\mbox{\scriptsize $(4)$}   &                                   \mergeG({\HH},{\KK},\ell^\to,\G, \pr\to\ps: \Set{\ell'_j.\G'_j\mid {\scriptstyle 1 \leq j \leq m }})  &=& 
      \pr\to\ps : \Set{\ell'_j. \mergeG({\HH},{\KK},\ell^\to,\G,\G'_j)\mid {\scriptstyle 1 \leq j \leq m }}\\[1.5mm]
\mbox{\scriptsize $(5)$}   &       \mergeG({\HH},{\KK},\noint ,\G, \pr\to\KK: \Set{\ell'_j.\G'_j\mid {\scriptstyle 1 \leq j \leq m }}) &=& 
      \pr\to\KK : \Set{\ell'_j. {\mergeG({\HH},{\KK},{\ell'_j}^\from,\G,\G'_j) } \mid  {\scriptstyle 1 \leq j \leq m }}\\[1.5mm]
\mbox{\scriptsize $(6)$}   &   \mergeG({\HH},{\KK},\ell^\from, \HH\to\pq: \Set{\ell_i.\G_i\mid {\scriptstyle 1 \leq i \leq n }},\G')  &=&  \KK\to\HH: \ell. \HH\to\pq : \ell. \mergeG({\HH},{\KK},\noint,\G_\vr,\G')\\
& & &       
\textit{if $\ell= \ell_\vr$ with $1 \leq \vr \leq n$}
\\[1.5mm]                                   
\mbox{\scriptsize $(7)$}   &                                     \mergeG({\HH},{\KK},\ell^\from, \pp\to\pq: \Set{\ell_i.\G_i\mid {\scriptstyle 1 \leq i \leq n }},\G')  &=& 
      \pp\to\pq : \Set{\ell_i.\mergeG({\HH},{\KK},\ell^\from,\G_i,\G') \mid {\scriptstyle 1 \leq i \leq n }} \\[1.5mm]
\mbox{\scriptsize $(8)$}   &        \mergeG({\HH},{\KK},\noint,{\Gvti \pp\pq \ell \S {\G}}, \G')  &=& 
      \pp\to\pq : \Set{\ell_i. \mergeG(\KK,\HH,\noint, \G', \G_i) \mid {\scriptstyle 1 \leq i \leq n } }\\[1.5mm]  
      \mbox{\scriptsize $(9)$}   &       \mergeG({\HH},{\KK},\noint ,\G, \pr\to\ps: \Set{\ell'_j.\G'_j\mid {\scriptstyle 1 \leq j \leq m }}) &=& 
      \pr\to\ps: \Set{\ell'_j.\mergeG({\KK},{\HH},\noint, \G'_j ,\G)\mid {\scriptstyle 1 \leq j \leq m }}   
      \\[1.5mm]                              
\end{array}
$}
\end{definition}

The argument `$\ell^\to$' (resp. `$\ell^\from$')  in $\mergeG({\HH},{\KK},\ell^\to, \G,\G')$
(resp. $\mergeG({\HH},{\KK},\ell^\from, \G,\G')$) is used when
a sending of the message `$\ell$' from 
$\KK$ (resp. $\HH$) 
is expected in the second (resp. first) global type in the subsequent recursive calls.\\
The argument $'\noint'$ is  used instead when all other possible interactions can occur in either the first or
the second global type in the subsequent recursive calls.\\
 In global types, the order of interactions between pairs of unrelated participants is irrelevant, since we would get the very same projections. In clauses {\scriptsize $(8)$} and {\scriptsize $(9)$}, however, we swap roles $\HH$ and $\KK$, as well as their corresponding global types in the ``recursive call".
We do that in order to avoid that in $\mergeG({\HH},{\KK},\noint, \G,\G')$ the interactions preceding a communication via gateway all belong to $\G$ (or $\G'$) and that the communication is completed after the description of interactions all belonging to  $\G'$ (or $\G$).  In this way the parallel nature of the interactions in $\G$ and $\G'$ that are not affected by the communications via gateways is made visually more evident. 

\begin{example}\mylabel{ex:f}
{\em
The protocol implemented by the multiparty session $\M'$  defined on page \pageref{smp}  
can be represented by the following global type $\G_\pr$:
 $$
 \begin{array}{cc}
\begin{array}{l}
 \G_\pr  =    \KK \to \pr : \textit{text}. \G'_\pr \\[1mm]
 \G'_\pr   =     \pr \to \KK: \{ \textit{ack}.  \pr \to \ps: \textit{go}.  \G_\ps,\\
      \phantom{\pq \to \KK: \{    }     \textit{nack}. \pr\to\ps: \textit{wait}. \KK\to \pr : \textit{transf}.  \G'_\pr   \}
\end{array}
&
\begin{array}{l}
\G_\ps   =    \KK \to \ps :  \textit{text}. \G'_\ps \\[1mm]
 \G'_\ps  =      \ps \to \KK: \{ \textit{ack}.  \ps \to \pr: \textit{go}.  \G_\pr,\\
      \phantom{\pq \to \KK: \{    }     \textit{nack}. \ps\to\pr: \textit{wait}.  \KK\to \ps : \textit{transf}. \G'_\ps \}
\end{array}
\end{array}
$$
Then, by Definition \ref{d:gc}, the composition, via $\HH$ and $\KK$, of the $\G$ of Example  \ref{ex:G} and the above $\G_\pr$ is:
$$
\begin{array}{lll}
\G\conn{\HH}{\KK}\G_\pr & =  &    \pp \to \q :   \textit{text} .  \q \to \HH  : \textit{text} . 
\HH \to \KK  : \textit{text} .\KK \to \pr : \textit{text}.   \G'_\pr\conn{\KK}{\HH}\G_1\\[2mm]
\G'_\pr\conn{\KK}{\HH}\G_1 & =  & \pr \to \KK: \{ \textit{ack}. \KK \to \HH  : \textit{ack} . \HH \to \q  : \textit{ack} . \pr \to \ps: \textit{go}.  
\q \to \pp : \textit{ok}. \G_\ps\conn{\KK}{\HH}\G,\\
&& \phantom{\pq \to \KK: \{  }       \textit{nack}. \KK \to \HH  : \textit{nack} .\HH \to \q :  \textit{nack}.\pr\to\ps: \textit{wait}.\q \to \pp : \textit{notyet}. \\
& &\phantom{\pq \to \KK: \{  \textit{nack}.  }\q\to\HH: \textit{transf}.\HH\to\KK: \textit{transf} .\KK\to  \pr  : \textit{transf}.           
 \G'_\pr\conn{\KK}{\HH}\G_1  \}\\[2mm]
\G_\ps\conn{\KK}{\HH}\G & =  & \pp \to \q :   \textit{text} .   \q \to \HH  : \textit{text} . 
\HH \to \KK  : \textit{text} .
\KK \to \ps : \textit{text}.   \G_1\conn{\HH}{\KK}\G'_\ps\\[2mm]
 \G_1\conn{\HH}{\KK}\G'_\ps  & =  &\ps \to \KK: \{ \textit{ack}. \KK \to \HH  : \textit{ack} . \HH \to \q  : \textit{ack} .
\q \to \pp : \textit{ok}. \ps \to \pr: \textit{go}.  \G\conn{\HH}{\KK}\G_\pr,\\
&& \phantom{\pq \to \KK: \{    }        \textit{nack}.   \KK \to \HH  : \textit{nack} .\HH \to \q :  \textit{nack}.\q \to \pp : \textit{notyet}.\ps\to\pr: \textit{wait}. \\
&& \phantom{\pq \to \KK: \{  \textit{nack}.  }\q\to\HH: \textit{transf}.\HH\to\KK: \textit{transf} .\KK\to  \ps  : \textit{transf}.           
 \G_1\conn{\HH}{\KK}\G'_\ps   \}
\end{array}
$$
In  $\G\conn{\HH}{\KK}\G_\pr$ 
 the text messages coming from $\pp$ are delivered to $\q$ and, alternately,
to $\pr$ and $\ps$ 
till they are  accepted ({\em ack}).  Participant  $\pp$ is informed when
text messages are accepted ({\em ok}). During the  cycle,  $\pq$ 
transforms a not yet accepted text into a more suitable form. The messages between $\pq$ and $\pr$ 
and $\ps$ are  exchanged by passing through the coupled forwarders $\HH$ and  $\KK$.

It is worth pointing out that 
 in  $\G\conn{\HH}{\KK}{\G_\pr}$, 
the {\em stop} branch of $\G$  disappeared. 
In fact, since any message coming from $\HH$ in $\G$ does now come from $\KK$ (which is now the gateway forwarding the messages coming in turn from either $\pr$ or $\ps$),
the function $\mergeG$ takes care of the fact that only
{\em ack} or {\em nack} can be received by (the gateway) $\HH$. 
This fact is reflected in the following
Theorem \ref{lem:mainlem}, where it is shown that the projections on $\HH$ and $\KK$ of $\G\conn{\HH}{\KK}{\G}'$
are a ``supertype" of $\gateway{\proj{\G}{\HH},\KK}$  and   
$\gateway{\proj{\G'}{\KK},\HH}$, respectively.

 We  could 
 look  at both $\HH$ and $\pp$ as interfaces:  $\HH$ representing a social-network system, which does not accept rude language, and  $\pp$
a social-network client sending text messages and requiring to be informed about their delivery status.  
From this point of view, the global type $\G$ of Example  \ref{ex:G} actually
describes a ``delivery-guaranteed'' service for text messages, assuring messages to be eventually delivered
 by means of a text-transformation policy. 
 }
\end{example}

 The following lemma assures that the global types obtained 
 during the evaluation of $\mergeG$ are always compatible.  

\begin{lemma}\mylabel{pc} 
Let $\pair\G\HH\interfacecomp\pair{\G'}\KK$.
Then for any call in the tree
of the recursive calls of  $\mergeG({\HH},{\KK},\noint,\G, \G')$:
\begin{enumerate}[(a)]
\item\mylabel{a} if  the call is $\mergeG({\HH},{\KK},\noint,\Y, \Y')$, then  $\proj{\Y}{\HH}\ \interfacecomp\ \proj{\Y'}{\KK}$;
\item\mylabel{b} if the call is $\mergeG({\HH},{\KK},\ell^\to,\Y,\Y')$, then  $\pin {\pp}{\ell.\proj{\Y}{\HH}}\ \interfacecomp\ \proj{\Y'}{\KK}$  for some $\pp$; 
\item\mylabel{c} if the call is $\mergeG({\HH},{\KK},\ell^\from,\Y,\Y')$, then  $\proj{\Y}{\HH}\ \interfacecomp\ \pin {\pp}{\ell.\proj{\Y'}{\KK}}$  for some $\pp$.  
\end{enumerate}
\end{lemma}
\begin{myproof} We show \ref{a}, \ref{b} and \ref{c} simultaneously by induction on the depth of the call in the tree,
 and by cases on the applied rule. For rule (1) the proof is immediate, since no new call is generated.\\
Rule (2). By induction  on \ref{a},  $\proj{( \Gvti \pp\HH \ell \S {\Y})}\HH \ \interfacecomp\ \proj{\Y'}\KK$. By Definition \ref{definition:projection}\\ \centerline{$\proj{( \Gvti \pp\HH \ell \S {\Y})}\HH=\pin \pp{\setlp{\ell}{\Y}{n}{\HH}}$} Then $\pin \pp{\setlp{\ell}{\Y}{n}{\HH}}\ \interfacecomp\ \proj{\Y'}\KK$, which implies $\pin\pp{\ell_i.\proj{\Y_i}\HH} \ \interfacecomp\ \proj{\Y'}\KK$ for $1\leq i\leq n$ by Proposition \ref{p:gw}\ref{p:gw1}.\\
Rule (3). By induction  on \ref{b},  $\pin \pp{\ell.\proj\Y\HH}\ \interfacecomp\ \proj{(\KK\to\ps: \Set{\ell'_j.\Y'_j\mid {\scriptstyle 1 \leq j \leq m }})}\KK$. By Definition \ref{definition:projection}\\ \centerline{$\proj{(\KK\to\ps: \Set{\ell'_j.\Y'_j\mid {\scriptstyle 1 \leq j \leq m }})}\KK=\pout\ps{\Set{\ell'_j.\proj{\Y'_j}\KK\mid {\scriptstyle 1 \leq j \leq m }}}$}
By Proposition \ref{p:gw}\ref{p:gw2} $\ell=\ell'_\vr$ with $1\leq\vr\leq m$ implies $\proj\Y\HH\ \interfacecomp \ \proj{\Y'_\vr}\KK$.\\
Rule (4). By induction  on \ref{b},  $\pin \pp{\ell.\proj\Y\HH}\ \interfacecomp\ \proj{(\pr\to\ps: \Set{\ell'_j.\Y'_j\mid {\scriptstyle 1 \leq j \leq m }}) }\KK$. By Definition \ref{def:compatibility} the projection\\ \centerline{$ \proj{(\pr\to\ps: \Set{\ell'_j.\Y'_j\mid {\scriptstyle 1 \leq j \leq m }}) }\KK$} must be an output, which implies $ \proj{(\pr\to\ps: \Set{\ell'_j.\Y'_j\mid {\scriptstyle 1 \leq j \leq m }}) }\KK= \proj{\Y'_1}\KK$ and $\proj{\Y'_j}\KK=\proj{\Y'_l}\KK$ for $1\leq j,l\leq m$ by Definition \ref{definition:projection}. We conclude $\pin \pp{\ell.\proj\Y\HH}\ \interfacecomp\ \proj{\Y'_j}\KK$ for $1 \leq j \leq m$.\\
Rule (8). By induction  on \ref{a},  $\proj{(\Gvti \pp\pq \ell \S {\Y})}\HH  \ \interfacecomp\ \proj{\Y'}\KK$.  By Definition \ref{definition:projection} either\\ \centerline{$\proj{(\Gvti \pp\pq \ell \S {\Y})}\HH =\proj{\Y_1}\HH$} and $\proj{\Y_i}\HH=\proj{\Y_l}\HH$ for $1\leq i,l\leq n$ or $\proj{(\Gvti \pp\pq \ell \S {\Y})}\HH = \pin\pt {(\mklab{\Lab_1}{\mklab{\ldots}{\Lab_n}})}$ and $\proj{\Y_i}\HH=\pin\pt{\Lab_i}$ for $1\leq i\leq n$. In the first case we get immediately $\proj{\Y_i}\HH\ \interfacecomp\ \proj{\Y'}\KK$ for $1\leq i\leq n$. In the second case by Proposition \ref{p:gw}\ref{p:gw1} $\pin\pt {(\mklab{\Lab_1}{\mklab{\ldots}{\Lab_n}})} \ \interfacecomp\ \proj{\Y'}\KK$ implies $\pin\pt{\Lab_i} \ \interfacecomp\ \proj{\Y'}\KK$, i.e. $\proj{\Y_i}\HH \ \interfacecomp\ \proj{\Y'}\KK$,  for $1\leq i\leq n$.\\
The proofs for rules (5), (6),(7) and  (9) are similar to those of rules (2), (3),(4) and (8), respectively.
\QED 
\end{myproof}

\medskip
Using the previous lemma we can show the soundness of Definition \ref{d:gc}. 
\begin{lemma}\mylabel{key} 
Let $\pair\G\HH\interfacecomp\pair{\G'}\KK$. Then
$\mergeG({\HH},{\KK},\noint,\G,\G')$ is defined and it is  a global type, i.e.
 a regular pre-global  type. 
\end{lemma}
\begin{myproof} 
 To show that $\mergeG({\HH},{\KK},\noint,\G,\G')$ is 
defined,
let us assume, towards a contradiction,  that in the tree of the recursive calls of
$\mergeG({\HH},{\KK},\noint,\G,\G')$ there is one leaf on which no rule of Definition
\ref{d:gc} can be applied.\\ 
If the recursive call is $\mergeG({\HH},{\KK},\noint,\Y,\Y')$, then the applicable rules are (1), (2), (5),  (8) and (9).  So the only deadlock would be for $\Y=\HH\to\pp:\Labgt$ and $\Y'\not=\pr\to\KK:\Labgt'$  and $\Y'\not=\pr\to\ps:\Labgt''$.  This is impossible since by Lemma \ref{pc}\ref{a} $\proj\Y\HH\ \interfacecomp\ \proj{\Y'}\KK$. 
If the recursive call is $\mergeG({\HH},{\KK},\ell^\to,\Y,\Y')$, then the applicable rules are (3) and (4). So the only deadlock would be for $\Y'\not=\pr\to\ps:\Labgt$ and $\Y'\not=\pr\to\KK:\Labgt'$. This is impossible since by Lemma \ref{pc}\ref{b} $\pin\pp\ell.\proj\Y\HH\ \interfacecomp\ \proj{\Y'}\KK$. The proof for the recursive call $\mergeG({\HH},{\KK},\ell^\from,\Y,\Y')$  uses Lemma \ref{pc}\ref{c}  and  it  is 
similar to the previous one.\\ 
 The regularity of the obtained pre-global type follows from observing that 
the regularity of $\G$ and $\G'$  forbid  
an infinite path, in the tree of the recursive calls, in which no two calls  are identical. \QED
\end{myproof}

\noindent
We can now prove the main result concerning projections of 
types  obtained by connecting via gateways. 

\begin{theorem}
\mylabel{lem:mainlem}
 If $\pair\G\HH\interfacecomp\pair{\G'}\KK$, then $\G\conn{\HH}{\KK}\G'$ is well formed.
 Moreover
\begin{enumerate}[label=(\roman*)]
\item
\mylabel{mainlem-i}
 $\gateway{\proj\G\HH,\KK}\subt\proj{(\G\conn{\HH}{\KK}\G' )}{\HH}$ and $\gateway{\proj{\G'}\KK,\HH}\subt\proj{(\G\conn{\HH}{\KK}\G' )}{\KK}$; 
\item 
\mylabel{mainlem-ii}
 $\proj{\G}{\pp}\subt\proj{(\G\conn{\HH}{\KK}\G' )}{\pp} $  and $\proj{\G'}{\pq}\subt\proj{(\G\conn{\HH}{\KK}\G' )}{\pq}$ ,  \\ for any $\pp\in\participantG{\G}$ and  $\pq\in\participantG{\G'}$ such that $\pp\neq\HH$ and $\pq\neq\KK$.
\end{enumerate}
\end{theorem}
\begin{myproof}
It is easy to verify that if\\
\centerline{$w=max\set{\depth\G\pp\mid \pp\in\participantG\G}$ and $w'=max\set{\depth{\G'}\pp\mid \pp\in\participantG{\G'}}$} then $\depth{\G\conn{\HH}{\KK}\G' }\pp\leq 2(w+w')$ for all $\pp\in\participantG\G\cup\participantG{\G'}$. Since \ref{mainlem-i} and \ref{mainlem-ii} imply that 
$\G\conn{\HH}{\KK}\G'$ is projectable for all $\pp\in\participantG\G\cup \participantG{\G'}$, then $\G\conn{\HH}{\KK}\G'$ is well formed.  Let $\star\in\set{\noint,\ell^\to,\ell^\from}$. \\
\ref{mainlem-i}. We only show  $\gateway{\proj\G\HH,\KK}\subt\proj{(\G\conn{\HH}{\KK}\G' )}{\HH}$, the proof of  $\gateway{\proj{\G'}\KK,\HH}\subt\proj{(\G\conn{\HH}{\KK}\G' )}{\KK}$  is specular.
We  prove 
that, for any recursive call $\mergeG({\HH},{\KK},\star,\Y, \Y')$ in $\mergeG({\HH},{\KK},\noint,\G, \G')$, the following relations between processes hold:
\begin{enumerate}[(a)]
\item \mylabel{ma}$\gateway{\proj\Y\HH,\KK}\subt\proj{\mergeG({\HH},{\KK},\noint,\Y, \Y')}\HH$; 
\item \mylabel{mb}$\pout\KK\ell.\gateway{\proj\Y\HH,\KK}\subt\proj{\mergeG({\HH},{\KK},\ell^\to,\Y, \Y')}\HH$;
\item \mylabel{mc} $\gateway{\proj{\Y}\HH,\KK}\subt\proj{\mergeG({\HH},{\KK},\ell^\from, \Y,\Y')}\HH$.
\end{enumerate}
We prove \ref{ma}, \ref{mb} and \ref{mc} simultaneously by coinduction on $\Y$ and $\Y'$ and by cases on the rule applied to get $\G\conn{\HH}{\KK}\G' =\mergeG({\HH},{\KK},\noint,\G, \G')$. Rules (1), (4), (5), (7),  (8)  and (9)  do not modify the communications of participant $\HH$, so coinduction easily applies.\\
Rule (2): 
$\mergeG({\HH},{\KK},\noint,{\Gvti \pp\HH \ell \S {\Y}}, \Y') =
      \pp\to\HH : \Set{\ell_i. {\mergeG({\HH},{\KK},\ell_i^\to,\Y_i,\Y') } \mid  {\scriptstyle 1 \leq i \leq n }}$.\\
       Let $\Y=\Gvti \pp\HH \ell \S {\Y}$,  then $\proj\Y\HH=\pin\pp\Set{\ell_i. \proj{\Y_i}{\HH} \mid  {\scriptstyle 1 \leq i \leq n }}$ by Definition \ref{definition:projection}.\\  
     \centerline{$\begin{array}{llll}
     \gateway{\proj\Y\HH,\KK}&=&\pin\pp\Set{\ell_i. \pout\KK{\ell_i}.\gateway{\proj{\Y_i}{\HH},\KK} \mid  {\scriptstyle 1 \leq i \leq n }}&\text{by Definition \ref{def:gc}}\\
     &\subt&\pin\pp\Set{\ell_i. \proj{\mergeG({\HH},{\KK},\ell_i^\to,\Y_i,\Y') }\HH \mid  {\scriptstyle 1 \leq i \leq n }}&\text{by rule $\rulename{sub-in}$ of Definition \ref{definition:subt} since}\\
     &&&\text{$\pout\KK{\ell_i}.\gateway{\proj{\Y_i}\HH ,\KK}\subt\proj{\mergeG({\HH},{\KK},\ell_i^\to,\Y_i,\Y') }\HH $}\\
     &&&\text{ for $1 \leq i \leq n$ by coinduction on \ref{mb} }\\
     &=&\proj{( \pp\to\HH : \Set{\ell_i. {\mergeG({\HH},{\KK},\ell_i^\to,\Y_i,\Y') } \mid  {\scriptstyle 1 \leq i \leq n }})}\HH&\text{by Definition \ref{definition:projection}} 
     \end{array}$}
   \\    
 Rule (3):$ \mergeG({\HH},{\KK},\ell^\to,\Y, \KK\to\ps: \Set{\ell'_j.\Y'_j\mid {\scriptstyle 1 \leq j \leq m }})  = \HH\to\KK: \ell. \KK\to\ps : \ell. \mergeG({\HH},{\KK},\noint,\Y,\Y'_\vr)$, where $\ell= \ell'_\vr$ with $1 \leq \vr \leq m$.\\
   \centerline{$\begin{array}{llll}
 \pout\KK\ell.\gateway{\proj\Y\HH,\KK}&\subt&\pout\KK\ell.\proj{\mergeG({\HH},{\KK},\noint,\Y,\Y'_\vr)}\HH&\text{by rule $\rulename{sub-out}$ of Definition \ref{definition:subt} since}\\
 &&&\gateway{\proj\Y\HH,\KK}\subt\proj{\mergeG({\HH},{\KK},\noint,\Y,\Y'_\vr)}\HH\\
 &&&\text{by coinduction on \ref{ma}}\\
 &=&\proj{(\HH\to\KK: \ell. \KK\to\ps : \ell. \mergeG({\HH},{\KK},\noint,\Y,\Y'_\vr))}\HH&\text{by Definition \ref{definition:projection}}
      \end{array}$} 
 Rule (6): $\mergeG({\HH},{\KK},\ell^\from, \HH\to\pq: \Set{\ell_i.\Y_i\mid {\scriptstyle 1 \leq i \leq n }},\Y')  =  \KK\to\HH: \ell. \HH\to\pq : \ell. \mergeG({\HH},{\KK},\noint,\Y_\vr,\Y')$, where $\ell= \ell_\vr$ with $1 \leq \vr \leq n$. Let $\Y=\Gvti \HH\q \ell \S \Y$,  then $\proj\Y\HH=\pout\pq\Set{\ell_i. \proj{\Y_i}{\HH} \mid  {\scriptstyle 1 \leq i \leq n }}$ by Definition \ref{definition:projection}.\\  
  \centerline{$\begin{array}{llll}
  \gateway{\proj\Y\HH,\KK}&=&\pin\KK\Set{\ell_i. \pout\q{\ell_i}.\gateway{\proj{\Y_i }\HH,\KK} \mid  {\scriptstyle 1 \leq i \leq n }}&\text{by Definition \ref{def:gc}}\\
&\subt&\pin\KK\ell.\pout\pq\ell.\gateway{\proj{\Y_\vr}\HH,\KK}&\text{by rule $\rulename{sub-in}$ 
and } \ell=\ell_\vr\\
&\subt&\pin\KK\ell.\pout\pq\ell.\proj{(\mergeG({\HH},{\KK},\noint,\Y_\vr,\Y'))}\HH&\text{by rules $\rulename{sub-in}$ and $\rulename{sub-out}$ since
}\\
&&&
\gateway{\proj{\Y_\vr}\HH,\KK}\subt\proj{\mergeG({\HH},{\KK},\noint,\Y_\vr,\Y')}\HH\\
&&&\text{by coinduction on \ref{ma}}\\
&=&\proj{(\KK\to\HH: \ell. \HH\to\pq : \ell. \mergeG({\HH},{\KK},\noint,\Y_\vr,\Y'))}\HH&\text{by Definition \ref{definition:projection}}
\end{array}$}  
 \ref{mainlem-ii}. We only show  $\proj\G\q\subt\proj{(\G\conn{\HH}{\KK}\G')}\q$  for $\q\in\participantG\G$ and $\q\not=\HH$. 
  The proof of  $\proj\G\ps\subt\proj{(\G\conn{\HH}{\KK}\G')}\ps$  for $\ps\in\participantG{\G'}$ and $\ps\not=\KK$ is specular. 
 We consider the recursive calls $\mergeG({\HH},{\KK},\star,  \Y,\Y')$ in $\mergeG({\HH},{\KK},\noint,  \G,\G')$. 
  We prove $\proj\Y\q\subt \proj{\mergeG({\HH},{\KK},\star,\Y,\Y')}\q$ 
 by coinduction on $\Y,\Y'$ and by cases on the applied rule.
 The only rule which modifies the communications of $\q$ is rule (6).  Let $\Y=\HH\to\pq: \Set{\ell_i.\Y_i\mid {\scriptstyle 1 \leq i \leq n }}$,   then \\
 \centerline{$\begin{array}{llll}
  \proj\Y\q   & = & \pin\HH\Set{\ell_i.\proj{\Y_i}\q\mid {\scriptstyle 1 \leq i \leq n }}  &\text{by Definition \ref{definition:projection}}   \\
 &\subt&\pin\HH\ell. \proj{\Y_\vr}\q&\text{by rule $\rulename{sub-in}$ and }\ell=\ell_\vr\\
  &\subt& \pin\HH\ell. \proj{(\mergeG({\HH},{\KK},\noint,\Y_\vr,\Y'))}\q&\text{by rule $\rulename{sub-in}$ since by coinduction}\\ 
  &&&\proj{\Y_\vr}\q\subt\proj{\mergeG({\HH},{\KK},\noint,\Y_\vr,\Y')}\q\\
  &=&\proj{( \KK\to\HH: \ell. \HH\to\pq : \ell. \mergeG({\HH},{\KK},\noint,\Y_\vr,\Y'))}\q&\text{by Definition \ref{definition:projection}   }  
  \end{array}$}  
  \begin{flushright}\vspace{-4mm}\QED\end{flushright}
\end{myproof}
We now show that if we start from two well-typed sessions which are compatible, then by building their 
 connection  
via gateways we get a well-typed session too. 
This is relevant, since well-typed sessions enjoy lock-freedom (Theorem \ref{theorem:waitfreedom}).
\begin{theorem}\mylabel{t:m} If $\pair\M\HH\interfacecomp\pair{\M'}\KK$ and $\der{}\M\G$ and $\der{}{\M'}{\G'}$, then $\der{}{\M\conn{\HH}{\KK}\M' }{\G\conn{\HH}{\KK}\G' }$.
\end{theorem}
\begin{myproof}
The typing $\der{}\M\G$ implies $\participantG\G\subseteq\participantS\M$. The typing $\der{}{\M'}{\G'}$ implies $\participantG{\G'}\subseteq\participantS{\M'}$. Then $\participantS{\M}\cap\participantS{\M'}=\emptyset$ gives $\participantG{\G}\cap\participantG{\G'}=\emptyset$.
Let $\M=\M_1\pc\pa\HH H$ and $\M'=\M_2\pc\pa\KK K$. By construction\\ \centerline{$\M\conn{\HH}{\KK}\M' =\M_1\pc\M_2\pc\pa\HH {\gateway {H,\KK}}\pc\pa\KK {\gateway{K,\HH}}$} From $\der{}\M\G$ we get $H\subt\proj\G\HH$. From $\der{}{\M'}{\G'}$ we get $K\subt\proj{\G'}\KK$. Lemma \ref{lem:compsubt} implies 
$\proj\G\HH\interfacecomp\proj{\G'}\KK$.\linebreak
Lemma \ref{l:gp} implies 
$\gateway {H,\KK}\subt\gateway {\proj\G\HH,\KK}$ and $\gateway {K,\HH}\subt\gateway {\proj{\G'}\KK,\HH}$. \\
We conclude  $\der{}{\M\conn{\HH}{\KK}\M' }{\G\conn{\HH}{\KK}\G' }$  using the projections of $\G\conn{\HH}{\KK}\G' $ given in Theorem \ref{lem:mainlem}. \QED
\end{myproof}

It is worth noticing that $\pair\G\HH\interfacecomp\pair{\G'}\KK$ and $\der{}\M\G$ and $\der{}{\M'}{\G'}$  do not imply $\pair\M\HH\interfacecomp\pair{\M'}\KK$. \\
Take as example $\M=\pa\pp{\pin\HH\ell}\pc\pa\HH{\pout\pp\ell}$, $\M'=\pa\q{\pout\KK\ell}\pc\pa\KK{\pin\q{\set{\ell,\ell'}}}$, $\G=\HH\to\pp:\ell$, $\G'=\q\to\KK:\ell$. In fact $\pout\pp\ell\interfacecomp\pin\q{\set{\ell,\ell'}}$ does not hold.
\begin{remark}\label{rem:counterexs}
{\em The proof of Theorem \ref{t:m} uses Lemma \ref{l:gp} which fails for the typing system $\vdash^+$. In spite of this, we conjecture that Theorem \ref{t:m} holds for $\vdash^+$ as well. The main reason is that compatibility requires all inputs  to have corresponding outputs, and this forbids to exploit the difference between $\subt$ and $\subtp$.
 }\end{remark}
Of course we could relax our typing system so that, in Rule $\rulename{t-sess}$,
$\subt$ is used 
for  interface processes  (i.e. those that are transformed into gateways when systems are connected), while 
$\subtp$ is used for all other processes. 
This would result, however, in a fairly serious restriction of the flexibility of  system connections, since we should 
establish a priori the interfaces of  systems.   \\

 As a possible general applications of our results,
let us suppose we have two  systems that correspond to  multiparty-sessions that are compatible via some participants (according to Definition \ref{msc}) and that are well typed (according to Definition \ref{definition:typing}). At this point we can ``deploy" the connected system (following Definition \ref{def:mptscomp}) without any 
 further 
verification step, since Theorem \ref{t:m} ensures  that  in such conditions we have a well-typed  and hence lock-free  connected system. 
Besides, we are able to provide the documentation (the global type) of the resulting systems.

%% file: OpenGTsynch-discussion.tex

\mysection{Related Works}\mylabel{s:rfw}
The distinguishing feature of an {\em open} system of concurrent components is its 
capacity of communicating with the ``outside'', i.e.\ with an environment of the system.
This ability provides means for composing open systems to larger systems (which may still be open).
In order to  compose systems ``safely'', it is common practice to rely on interface descriptions. \\
 MPST systems  \cite{Honda2016,DBLP:conf/sfm/CoppoDPY15,MYHesop09} 
are usually assumed to be {\em closed}, since all the components needed for the functioning of the
system must be already there. 
In \cite{BdLH18} a novel approach to open systems has been proposed where, 
according to the current needs, the behaviour of any participant can be regarded
as an ``interface''.
An interface  is hence intended to represent - somehow dually with respect to the standard notion of interface -
part of  the expected communication behaviour of the environment.
Identifying a participant behaviour as interface corresponds to expecting such a behaviour to be realised by the environment rather than
by an actual component of the system. 
Then, according to such an approach, there is actually no distinction between 
a closed and an open system. 
In particular, 
once two systems possess two ``compatible''  interfaces, they can be connected.
The connecting mechanism  of  
 \cite{BdLH18} uses suitable forwarders, dubbed ``gateways'', 
  for  this purpose. 
The gateways are automatically synthesised out of the compatible interfaces  and the connection of two systems simply consists
in replacing the latter by the former.

In the present paper we have provided a multiparty formalism and we  have
adapted the approach of \cite{BdLH18} to it.
Our calculus of multiparty sessions is  like those 
of \cite{DezEtAlt16,GJPSY19}, but for the use of coinduction instead of induction which is inspired by \cite{CGP09,DS}.  As in \cite{DS} we get rid of local types, which in many calculi are similar to processes \cite{CDG19,DezEtAlt16,GJPSY19}.  The syntax of global types is the coinductive version of the standard syntax \cite{Honda2016} and the notion of projection is an extension of both the standard projection \cite{Honda2016} and the projection given in \cite{DS}.  Our global types  assure lock-freedom  of multiparty sessions.

A relevant feature of our formalism is that the connection operation on systems can be ``lifted'' to the level of global types. In  \cite{BdLH18}, where systems of  CFSMs   were taken into account,
 such a lifting was done by  extending the  syntax of global  descriptions with a {\em new} symbol, whose semantics is indeed the connection-by-gateways at system level.  Instead in the present paper we can use  the standard syntax to build the global type of the session obtained by connecting. 
Moreover, we have shown that  the compatibility relation of \cite{BdLH18}, which requires duality, can be relaxed to a relation
strongly  similar 
to  session-types'  subtyping \cite{GH05,DemangeonH11}. Our structural preorder  on processes mimics the subtyping relation between session types of  \cite{CDG19}, which is a restriction of the subtyping of \cite{DemangeonH11}. 
This choice is justified by the fact that the subtyping of \cite{DemangeonH11} allows process substitution, while the subtyping of \cite{GH05} allows channel substitution, as observed in \cite{Gay16}.

 In \cite{LT12}  global types are build out of several local types (under certain conditions). We also aim at
getting global types, the difference being that this is obtained out of the global types describing the two systems which are connected. The ``dynamic" addition of participants (they can join/leave the session after it's been set up) is supported in the calculus of \cite{HY17}. In that work the extension of a system is part of the global protocol. The extension operation is sort of  ``internalised".
We take instead the standard point of view of open systems, where the possible extensions cannot be 
``programmed"  in advance. 
The two approaches to the system-extension issue look orthogonal.

 Both ``arbiter processes''  \cite{CLMSW16} and ``mediums''  \cite{CP16} coordinate communications
described by global types. A difference  with the present paper  is that their aim is to reduce the interactions
in multiparty sessions to interactions in binary sessions.
Our gateways do instead act as simple ``forwarders", with the aim of connecting two multiparty systems.
Nonetheless, our work could be further developed and investigated
in the logical context of \cite{CLMSW16}: in the logical interpretation of multiparty sessions
one could introduce a ``connection-cut" corresponding to a sort of connection-by-gateways-operator. Then the good properties of the system corresponding to the proof containing the cut should be guaranteed
by proving that the ``connection-cut" is actually an admissible rule.
The proof should consist in a ``connection-cut elimination" procedure corresponding to our {\sc cn} function on global types, once extended
(as we claim it can be, see next Section) in order to ``bypass" the use of gateways.

\mysection{Future Works and Conclusion}\mylabel{s:fwc}
 The MPST framework does work fairly well for the design of closed systems, but does not possess the flexibility open systems can offer.
Managing to look at global types as overall descriptions of open systems results in the possibility of a modular design of systems.
From another point of view, by means of our approach one could develop systems where some participants, instead of representing
actual processes, describe sort of ``API calls", along the line of what some researchers refer to as Behavioural-API.
Moreover, the theory we propose could be helpful also after the system implementation phase. Let us assume to have a system developed
using the MPST software-development approach. After the implementation phase, one could realise that the service corresponding to a participant of the system can be more suitably provided by another system. The participant can then be safely replaced by a gateway connection with the other system and the connection operation on global types enables to get a global view of what is going on in the resulting system.

We conjecture the completeness of our process compatibility, i.e. that the session  $\pair\M\HH\interfacecomp\pair{\M'}\KK$ can reduce to a stuck session whenever $\HH$ and $\KK$ are not compatible. This could be shown by taking inspiration from the completeness proofs for subtyping of \cite{DezEtAlt16,GJPSY19}.

The use of gateways enables us to get a ``safe''  systems'  composition by minimally affecting the systems themselves, being just the interfaces to be modified. 
One could however wonder whether gateways are strictly necessary to get safe connections in our  multiparty-sessions'  setting.
 As suggested in \cite{LTPC}, one could try to ``bypass'' the use of gateways by taking the interface participants out 
and  changing   
some senders' and receivers' names in the other participants' ``code''.
The following simple example shows that just a renaming would not work in general.
Let us consider the following global types.\\
\centerline{$
\begin{array}{rcl@{\hspace{22mm}}rcl}
\G & = &\pp\to\HH:\ell.\G
&
\G' & = & \KK\to\pr:\ell.\KK\to\ps:\ell.\G'
\end{array}
$}
It is immediate to check that 
the multiparty sessions corresponding to $\G$ and $\G'$ are\\
\centerline{$
\begin{array}{rcl@{\hspace{22mm}}rcl}
\M & = &\pa \pp \PP \mid \pa \HH H 
&
\M'& = &\pa \KK K \mid \pa \pr R  \mid \pa \ps S
\end{array}
$}
where
$\quad
P =   \pout{\HH}{\ell}.P
\qquad
H  =   \pin{\pp}{\ell}.H
\qquad
K =  \pout{\pr}{\ell}{}.\pout{\ps}{\ell}.K
\qquad
R =  \pin{\KK}{\ell}.R
\qquad
S =  \pin{\KK}{\ell}.S
$\\
On  the side of $\M'$ we could take $K$ out and change some senders' and receivers' names
in $R$ and $S$ in order they can receive the message $\ell$ directly from $\pp$,  so  obtaining
$\tilde R = {\pin{\pp}{\ell}.\tilde R}$  and $\tilde S = {\pin{\pp}{\ell}. \tilde S}$.
On the side of $\M$, instead, after taking out $H$,
we could not get a sound connection by  a simple renaming for the recipient $\HH$ in $P = \pout{\HH}{\ell}.P$,
since the message $\ell$ should be delivered, alternately, to  $\tilde R$ and $\tilde S$.
A safe connection would hence imply also a modification of the ``code'' of $P$ as follows:  $\tilde\PP = \pout{\pr}{\ell}.\pout{\ps}{\ell}.\tilde\PP$.
We conjecture that the function $\;\,\conn{\HH}{\KK}\,\;$ on global types can be redefined in such a way
that, in the present example, by projecting
$\G\conn{\HH}{\KK}\G'$  we  exactly obtain 
$ \pa \pp {\tilde{\PP}}     \pc    \pa \pr {\tilde R}    \pc   \pa \ps {\tilde S}$.
This new connection operation on global types would result in a useful tool for the modular design of systems
via global types. We leave the investigation of such alternative definition of 
$\ \conn{\HH}{\KK}\; $ for future work.

 The results of this paper would be more applicable accounting for asynchronous communications. 
In particular, a first relevant step would consist in  allowing 
gateways to interact asynchronously.
We expect the compatibility could be extended, since the subtyping for asynchronous multiparty sessions is more permissive than
the subtyping for the synchronous ones \cite{MYHesop09}. Of course this extension requires care, being the subtyping of \cite{MYHesop09} undecidable,
as shown in  \cite{BCZ17,LY17}.

The connection via gateways proposed by the authors of \cite{BdLH18} and 
exploited in the present paper in a multiparty sessions  setting does produce networks of systems possessing a tree-like topology.
In order to get general graphs topology, it sounds natural to
extend  the present ``single interface'' connection  to a ``multiple interfaces'' one.
Such  an extension,
however,
immediately reveals itself to be unsound:
by connecting via gateways more than  one 
pair of
compatible interfaces one could obtain a deadlocked system. 
A very simple example  for that is 
 $\quad
\G = \pp\to\HH:\ell 
\quad\text{and}\quad
\G' = \KK\to\ps:\ell
$\\
By projection we get the systems
$\quad
\M = \pa \pp {\pout{\HH}{\ell}{}} \mid  \pa \HH {\pin{\pp}{\ell}{}}\quad\text{and}\quad
\M' =  \pa \KK {\pout{\ps}{\ell}{}} \mid \pa \ps {\pin{\KK}{\ell}{}} 
$\\
It is immediate to check that $\pp$ and $\ps$ are compatible, as well as $\HH$ and $\KK$.
Simultaneously connecting $\M$ and $\M'$ through both the compatible pairs ($\pp,\ps$) and ($\HH,\KK$)  would result in the following deadlocked  system
$\qquad\qquad
\pa \pp {\pin{\ps}{\ell}{}.\pout{\HH}{\ell}{}}\quad \mid\quad  \pa \HH {\pin{\pp}{\ell}{}.\pout{\KK}{\ell}{}}\quad
\mid\quad \pa \KK {\pin{\HH}{\ell}{}.\pout{\ps}{\ell}{}}\quad \mid\quad \pa \ps {\pin{\KK}{\ell}{}.\pout{\pp}{\ell}{}} 
$\\ 
(A similar example can be  developed 
also in the CFSMs setting of \cite{BdLH18}).
In order to guarantee ``safeness'' of multiple connections, suitable requirements have hence to be devised.
An adaptation to the present setting of the {\em interaction type system}  of \cite{DBLP:journals/mscs/CoppoDYP16} could be investigated in future for such an aim. 


%% file: OpenGTsynch-Main.bbl
\begin{thebibliography}{10}
\providecommand{\bibitemdeclare}[2]{}
\providecommand{\surnamestart}{}
\providecommand{\surnameend}{}
\providecommand{\urlprefix}{Available at }
\providecommand{\url}[1]{\texttt{#1}}
\providecommand{\href}[2]{\texttt{#2}}
\providecommand{\urlalt}[2]{\href{#1}{#2}}
\providecommand{\doi}[1]{doi:\urlalt{http://dx.doi.org/#1}{#1}}
\providecommand{\bibinfo}[2]{#2}

\bibitemdeclare{inproceedings}{BdLH18}
\bibitem{BdLH18}
\bibinfo{author}{Franco \surnamestart Barbanera\surnameend},
  \bibinfo{author}{Ugo \surnamestart de'Liguoro\surnameend} \&
  \bibinfo{author}{Rolf \surnamestart Hennicker\surnameend}
  (\bibinfo{year}{2018}): \emph{\bibinfo{title}{Global Types for Open
  Systems}}.
\newblock In: {\sl \bibinfo{booktitle}{ICE}}, {\sl \bibinfo{series}{EPTCS}}
  \bibinfo{volume}{279}, \bibinfo{publisher}{Open Publishing Association}, pp.
  \bibinfo{pages}{4--20}, \doi{10.4204/EPTCS.279.4}.

\bibitemdeclare{article}{BCZ17}
\bibitem{BCZ17}
\bibinfo{author}{Mario \surnamestart Bravetti\surnameend},
  \bibinfo{author}{Marco \surnamestart Carbone\surnameend} \&
  \bibinfo{author}{Gianluigi \surnamestart Zavattaro\surnameend}
  (\bibinfo{year}{2017}): \emph{\bibinfo{title}{{Undecidability of Asynchronous
  Session Subtyping}}}.
\newblock {\sl \bibinfo{journal}{Information and Computation}}
  \bibinfo{volume}{256}, pp. \bibinfo{pages}{300--320},
  \doi{10.1016/j.ic.2017.07.010}.

\bibitemdeclare{inproceedings}{CP16}
\bibitem{CP16}
\bibinfo{author}{Lu{\'{\i}}s \surnamestart Caires\surnameend} \&
  \bibinfo{author}{Jorge~A. \surnamestart P{\'{e}}rez\surnameend}
  (\bibinfo{year}{2016}): \emph{\bibinfo{title}{{Multiparty Session Types
  Within a Canonical Binary Theory, and Beyond}}}.
\newblock In: {\sl \bibinfo{booktitle}{FORTE}}, {\sl \bibinfo{series}{LNCS}}
  \bibinfo{volume}{9688}, \bibinfo{publisher}{Springer}, pp.
  \bibinfo{pages}{74--95}, \doi{10.1007/978-3-319-39570-8\_6}.

\bibitemdeclare{inproceedings}{CLMSW16}
\bibitem{CLMSW16}
\bibinfo{author}{Marco \surnamestart Carbone\surnameend}, \bibinfo{author}{Sam
  \surnamestart Lindley\surnameend}, \bibinfo{author}{Fabrizio \surnamestart
  Montesi\surnameend}, \bibinfo{author}{Carsten \surnamestart
  Sch{\"{u}}rmann\surnameend} \& \bibinfo{author}{Philip \surnamestart
  Wadler\surnameend} (\bibinfo{year}{2016}): \emph{\bibinfo{title}{{Coherence
  Generalises Duality: {A} Logical Explanation of Multiparty Session Types}}}.
\newblock In: {\sl \bibinfo{booktitle}{CONCUR}}, {\sl
  \bibinfo{series}{LIPIcs}}~\bibinfo{volume}{59}, \bibinfo{publisher}{Schloss
  Dagstuhl - Leibniz-Zentrum fuer Informatik}, pp.
  \bibinfo{pages}{33:1--33:15}, \doi{10.4230/LIPIcs.CONCUR.2016.33}.

\bibitemdeclare{incollection}{BookCC}
\bibitem{BookCC}
\bibinfo{author}{Felice \surnamestart Cardone\surnameend} \&
  \bibinfo{author}{Mario \surnamestart Coppo\surnameend}
  (\bibinfo{year}{2013}): \emph{\bibinfo{title}{{Recursive Types}}}.
\newblock In \bibinfo{editor}{Henk \surnamestart Barendregt\surnameend},
  \bibinfo{editor}{Wil \surnamestart Dekkers\surnameend} \&
  \bibinfo{editor}{Richard \surnamestart Statman\surnameend}, editors: {\sl
  \bibinfo{booktitle}{{Lambda Calculus with Types}}},
  \bibinfo{series}{Perspectives in Logic}, \bibinfo{publisher}{Cambridge
  University Press}, pp. \bibinfo{pages}{377--576},
  \doi{10.1017/CBO9781139032636.011}.

\bibitemdeclare{article}{CGP09}
\bibitem{CGP09}
\bibinfo{author}{Giuseppe \surnamestart Castagna\surnameend},
  \bibinfo{author}{Nils \surnamestart Gesbert\surnameend} \&
  \bibinfo{author}{Luca \surnamestart Padovani\surnameend}
  (\bibinfo{year}{2009}): \emph{\bibinfo{title}{{A Theory of Contracts for Web
  Services}}}.
\newblock {\sl \bibinfo{journal}{{ACM} Transactions on Programming Languages
  and Systems}} \bibinfo{volume}{31}(\bibinfo{number}{5}), pp.
  \bibinfo{pages}{19:1--19:61}, \doi{10.1145/1538917.1538920}.

\bibitemdeclare{article}{CDG19}
\bibitem{CDG19}
\bibinfo{author}{Ilaria \surnamestart Castellani\surnameend},
  \bibinfo{author}{Mariangiola \surnamestart Dezani-Ciancaglini\surnameend} \&
  \bibinfo{author}{Paola \surnamestart Giannini\surnameend}
  (\bibinfo{year}{2019}): \emph{\bibinfo{title}{{Reversible Sessions with
  Flexible Choices}}}.
\newblock {\sl \bibinfo{journal}{Acta Informatica}},
  \doi{10.1007/s00236-019-00332-y}.
\newblock \bibinfo{note}{To appear}.

\bibitemdeclare{inproceedings}{DBLP:conf/sfm/CoppoDPY15}
\bibitem{DBLP:conf/sfm/CoppoDPY15}
\bibinfo{author}{Mario \surnamestart Coppo\surnameend},
  \bibinfo{author}{Mariangiola \surnamestart Dezani{-}Ciancaglini\surnameend},
  \bibinfo{author}{Luca \surnamestart Padovani\surnameend} \&
  \bibinfo{author}{Nobuko \surnamestart Yoshida\surnameend}
  (\bibinfo{year}{2015}): \emph{\bibinfo{title}{{A Gentle Introduction to
  Multiparty Asynchronous Session Types}}}.
\newblock In: {\sl \bibinfo{booktitle}{Formal Methods for Multicore
  Programming}}, \bibinfo{series}{LNCS}, \bibinfo{publisher}{Springer}, pp.
  \bibinfo{pages}{146--178}, \doi{10.1007/978-3-319-18941-3\_4}.

\bibitemdeclare{article}{DBLP:journals/mscs/CoppoDYP16}
\bibitem{DBLP:journals/mscs/CoppoDYP16}
\bibinfo{author}{Mario \surnamestart Coppo\surnameend},
  \bibinfo{author}{Mariangiola \surnamestart Dezani{-}Ciancaglini\surnameend},
  \bibinfo{author}{Nobuko \surnamestart Yoshida\surnameend} \&
  \bibinfo{author}{Luca \surnamestart Padovani\surnameend}
  (\bibinfo{year}{2016}): \emph{\bibinfo{title}{{Global Progress for
  Dynamically Interleaved Multiparty Sessions}}}.
\newblock {\sl \bibinfo{journal}{Mathematical Structures in Computer Science}}
  \bibinfo{volume}{26}(\bibinfo{number}{2}), pp. \bibinfo{pages}{238--302},
  \doi{10.1017/S0960129514000188}.

\bibitemdeclare{article}{Courcelle83}
\bibitem{Courcelle83}
\bibinfo{author}{Bruno \surnamestart Courcelle\surnameend}
  (\bibinfo{year}{1983}): \emph{\bibinfo{title}{{Fundamental Properties of
  Infinite Trees}}}.
\newblock {\sl \bibinfo{journal}{{Theoretical Computer Science}}}
  \bibinfo{volume}{25}, pp. \bibinfo{pages}{95--169},
  \doi{10.1016/0304-3975(83)90059-2}.

\bibitemdeclare{inproceedings}{DemangeonH11}
\bibitem{DemangeonH11}
\bibinfo{author}{Romain \surnamestart Demangeon\surnameend} \&
  \bibinfo{author}{Kohei \surnamestart Honda\surnameend}
  (\bibinfo{year}{2011}): \emph{\bibinfo{title}{{Full Abstraction in a Subtyped
  Pi-Calculus with Linear Types}}}.
\newblock In: {\sl \bibinfo{booktitle}{CONCUR}}, {\sl \bibinfo{series}{LNCS}}
  \bibinfo{volume}{6901}, \bibinfo{publisher}{Springer}, pp.
  \bibinfo{pages}{280--296}, \doi{10.1007/978-3-642-23217-6\_19}.

\bibitemdeclare{inproceedings}{DezEtAlt16}
\bibitem{DezEtAlt16}
\bibinfo{author}{Mariangiola \surnamestart Dezani{-}Ciancaglini\surnameend},
  \bibinfo{author}{Silvia \surnamestart Ghilezan\surnameend},
  \bibinfo{author}{Svetlana \surnamestart Jaksic\surnameend},
  \bibinfo{author}{Jovanka \surnamestart Pantovic\surnameend} \&
  \bibinfo{author}{Nobuko \surnamestart Yoshida\surnameend}
  (\bibinfo{year}{2015}): \emph{\bibinfo{title}{{Precise Subtyping for
  Synchronous Multiparty Sessions}}}.
\newblock In: {\sl \bibinfo{booktitle}{PLACES}}, {\sl
  \bibinfo{series}{{EPTCS}}} \bibinfo{volume}{203}, \bibinfo{publisher}{Open
  Publishing Association}, pp. \bibinfo{pages}{29--43},
  \doi{10.4204/EPTCS.203.3}.

\bibitemdeclare{inproceedings}{Gay16}
\bibitem{Gay16}
\bibinfo{author}{Simon \surnamestart Gay\surnameend} (\bibinfo{year}{2016}):
  \emph{\bibinfo{title}{{Subtyping Supports Safe Session Substitution}}}.
\newblock In: {\sl \bibinfo{booktitle}{A List of Successes That Can Change the
  World - Essays Dedicated to Philip Wadler on the Occasion of His 60th
  Birthday}}, {\sl \bibinfo{series}{LNCS}} \bibinfo{volume}{9600},
  \bibinfo{publisher}{Springer}, pp. \bibinfo{pages}{95--108},
  \doi{10.1007/978-3-319-30936-1\_5}.

\bibitemdeclare{article}{GH05}
\bibitem{GH05}
\bibinfo{author}{Simon \surnamestart Gay\surnameend} \&
  \bibinfo{author}{Malcolm \surnamestart Hole\surnameend}
  (\bibinfo{year}{2005}): \emph{\bibinfo{title}{{Subtyping for Session Types in
  the Pi Calculus}}}.
\newblock {\sl \bibinfo{journal}{Acta Informatica}}
  \bibinfo{volume}{42}(\bibinfo{number}{2/3}), pp. \bibinfo{pages}{191--225},
  \doi{10.1007/s00236-005-0177-z}.

\bibitemdeclare{article}{GJPSY19}
\bibitem{GJPSY19}
\bibinfo{author}{Silvia \surnamestart Ghilezan\surnameend},
  \bibinfo{author}{Svetlana \surnamestart Jaksic\surnameend},
  \bibinfo{author}{Jovanka \surnamestart Pantovic\surnameend},
  \bibinfo{author}{Alceste \surnamestart Scalas\surnameend} \&
  \bibinfo{author}{Nobuko \surnamestart Yoshida\surnameend}
  (\bibinfo{year}{2019}): \emph{\bibinfo{title}{{Precise Subtyping for
  Synchronous Multiparty Sessions}}}.
\newblock {\sl \bibinfo{journal}{Journal of Logic and Algebraic Methods in
  Programming}} \bibinfo{volume}{104}, pp. \bibinfo{pages}{127--173},
  \doi{10.1016/j.jlamp.2018.12.002}.

\bibitemdeclare{inproceedings}{HVK98}
\bibitem{HVK98}
\bibinfo{author}{Kohei \surnamestart Honda\surnameend},
  \bibinfo{author}{Vasco~Thudichum \surnamestart Vasconcelos\surnameend} \&
  \bibinfo{author}{Makoto \surnamestart Kubo\surnameend}
  (\bibinfo{year}{1998}): \emph{\bibinfo{title}{{Language Primitives and Type
  Discipline for Structured Communication-Based Programming}}}.
\newblock In: {\sl \bibinfo{booktitle}{ESOP}}, {\sl \bibinfo{series}{LNCS}}
  \bibinfo{volume}{1381}, \bibinfo{publisher}{Springer}, pp.
  \bibinfo{pages}{122--138}, \doi{10.1007/BFb0053567}.

\bibitemdeclare{inproceedings}{HYC08}
\bibitem{HYC08}
\bibinfo{author}{Kohei \surnamestart Honda\surnameend}, \bibinfo{author}{Nobuko
  \surnamestart Yoshida\surnameend} \& \bibinfo{author}{Marco \surnamestart
  Carbone\surnameend} (\bibinfo{year}{2008}): \emph{\bibinfo{title}{{Multiparty
  Asynchronous Session Types}}}.
\newblock In: {\sl \bibinfo{booktitle}{POPL}}, \bibinfo{publisher}{ACM Press},
  pp. \bibinfo{pages}{273--284}, \doi{10.1145/1328438.1328472}.

\bibitemdeclare{article}{Honda2016}
\bibitem{Honda2016}
\bibinfo{author}{Kohei \surnamestart Honda\surnameend}, \bibinfo{author}{Nobuko
  \surnamestart Yoshida\surnameend} \& \bibinfo{author}{Marco \surnamestart
  Carbone\surnameend} (\bibinfo{year}{2016}): \emph{\bibinfo{title}{{Multiparty
  Asynchronous Session Types}}}.
\newblock {\sl \bibinfo{journal}{Journal of the ACM}}
  \bibinfo{volume}{63}(\bibinfo{number}{1}), p.~\bibinfo{pages}{9},
  \doi{10.1145/2827695}.

\bibitemdeclare{inproceedings}{HY17}
\bibitem{HY17}
\bibinfo{author}{Raymond \surnamestart Hu\surnameend} \&
  \bibinfo{author}{Nobuko \surnamestart Yoshida\surnameend}
  (\bibinfo{year}{2017}): \emph{\bibinfo{title}{{Explicit Connection Actions in
  Multiparty Session Types}}}.
\newblock In: {\sl \bibinfo{booktitle}{{FASE}}}, {\sl \bibinfo{series}{LNCS}}
  \bibinfo{volume}{10202}, \bibinfo{publisher}{Springer}, pp.
  \bibinfo{pages}{116--133}, \doi{10.1007/978-3-662-54494-5\_7}.

\bibitemdeclare{article}{Kobayashi02}
\bibitem{Kobayashi02}
\bibinfo{author}{Naoki \surnamestart Kobayashi\surnameend}
  (\bibinfo{year}{2002}): \emph{\bibinfo{title}{{A Type System for Lock-Free
  Processes}}}.
\newblock {\sl \bibinfo{journal}{Information and Computation}}
  \bibinfo{volume}{177}(\bibinfo{number}{2}), pp. \bibinfo{pages}{122--159},
  \doi{10.1006/inco.2002.3171}.

\bibitemdeclare{article}{KozenS17}
\bibitem{KozenS17}
\bibinfo{author}{Dexter \surnamestart Kozen\surnameend} \&
  \bibinfo{author}{Alexandra \surnamestart Silva\surnameend}
  (\bibinfo{year}{2017}): \emph{\bibinfo{title}{{Practical Coinduction}}}.
\newblock {\sl \bibinfo{journal}{Mathematical Structures in Computer Science}}
  \bibinfo{volume}{27}(\bibinfo{number}{7}), pp. \bibinfo{pages}{1132--1152},
  \doi{10.1017/S0960129515000493}.

\bibitemdeclare{misc}{LTPC}
\bibitem{LTPC}
\bibinfo{author}{Ivan \surnamestart Lanese\surnameend} \&
  \bibinfo{author}{Emilio \surnamestart Tuosto\surnameend}:
  \emph{\bibinfo{title}{Personal Communication}}.

\bibitemdeclare{inproceedings}{LT12}
\bibitem{LT12}
\bibinfo{author}{Julien \surnamestart Lange\surnameend} \&
  \bibinfo{author}{Emilio \surnamestart Tuosto\surnameend}
  (\bibinfo{year}{2012}): \emph{\bibinfo{title}{{Synthesising Choreographies
  from Local Session Types}}}.
\newblock In: {\sl \bibinfo{booktitle}{CONCUR}}, {\sl \bibinfo{series}{LNCS}}
  \bibinfo{volume}{7454}, \bibinfo{publisher}{Springer}, pp.
  \bibinfo{pages}{225--239}, \doi{10.1007/978-3-642-32940-1\_17}.

\bibitemdeclare{inproceedings}{LY17}
\bibitem{LY17}
\bibinfo{author}{Julien \surnamestart Lange\surnameend} \&
  \bibinfo{author}{Nobuko \surnamestart Yoshida\surnameend}
  (\bibinfo{year}{2017}): \emph{\bibinfo{title}{{On the Undecidability of
  Asynchronous Session Subtyping}}}.
\newblock In: {\sl \bibinfo{booktitle}{FOSSACS}}, {\sl \bibinfo{series}{LNCS}}
  \bibinfo{volume}{10203}, \bibinfo{publisher}{Springer}, pp.
  \bibinfo{pages}{441--457}, \doi{10.1007/978-3-662-54458-7\_26}.

\bibitemdeclare{inproceedings}{MYHesop09}
\bibitem{MYHesop09}
\bibinfo{author}{Dimitris \surnamestart Mostrous\surnameend},
  \bibinfo{author}{Nobuko \surnamestart Yoshida\surnameend} \&
  \bibinfo{author}{Kohei \surnamestart Honda\surnameend}
  (\bibinfo{year}{2009}): \emph{\bibinfo{title}{{Global Principal Typing in
  Partially Commutative Asynchronous Sessions}}}.
\newblock In: {\sl \bibinfo{booktitle}{ESOP}}, {\sl \bibinfo{series}{LNCS}}
  \bibinfo{volume}{5502}, \bibinfo{publisher}{Springer}, pp.
  \bibinfo{pages}{316--332}, \doi{10.1007/978-3-642-00590-9\_23}.

\bibitemdeclare{book}{pier02}
\bibitem{pier02}
\bibinfo{author}{Benjamin~C. \surnamestart Pierce\surnameend}
  (\bibinfo{year}{2002}): \emph{\bibinfo{title}{{Types and Programming
  Languages}}}.
\newblock \bibinfo{publisher}{MIT Press}.

\bibitemdeclare{article}{DS}
\bibitem{DS}
\bibinfo{author}{Paula \surnamestart Severi\surnameend} \&
  \bibinfo{author}{Mariangiola \surnamestart Dezani-Ciancaglini\surnameend}
  (\bibinfo{year}{2019}): \emph{\bibinfo{title}{{Observational Equivalence for
  Multiparty Sessions}}}.
\newblock {\sl \bibinfo{journal}{Fundamenta Informaticae}}
  \bibinfo{volume}{167}.
\newblock \bibinfo{note}{To appear}.

\end{thebibliography}
